%% file: 0main.tex
\documentclass{article}
\usepackage{ijcai26}

\usepackage{times}
\usepackage{soul}
\usepackage{url}
\usepackage[hidelinks]{hyperref}
\usepackage[utf8]{inputenc}
\usepackage[small]{caption}
\usepackage{graphicx}
\usepackage{algorithm}
\usepackage[switch]{lineno}

\usepackage{algpseudocode}
\usepackage{amsfonts}
\usepackage{amsthm,amsmath,amssymb}
\usepackage{thmtools}
\usepackage{thm-restate}
\usepackage{pifont}
\usepackage{multirow}
\usepackage{booktabs}
\usepackage{xcolor}

\usepackage{subfigure}
\usepackage{diagbox}
\usepackage{dblfloatfix}

\input{math_command.tex}

\title{Scalable Fair Influence Blocking Maximization via Approximately Monotonic Submodular Optimization}

\author{
Qiangpeng Fang$^1$
\and
Jilong Shi$^1$
\and
Xiaobin Rui$^{1,2}$
\and
Jian Zhang$^{1,2}$
\And
Zhixiao Wang$^{1,2,*}$
\affiliations
$^1$China University of Mining and Technology, Xuzhou 221116, China\\
$^2$Mine Digitization Engineering Research Center of the Ministry of Education, Xuzhou 221116, China\\ 
\emails
\{fangqp,jlshi,ruixiaobin,zhangjian10231209,zhxwang\}@cumt.edu.com
}

\begin{document}

\maketitle

\begin{abstract}
Influence Blocking Maximization (IBM) aims to select a positive seed set to suppress the spread of negative influence.
However, existing IBM methods focus solely on maximizing blocking effectiveness, overlooking fairness across communities.
To address this issue, we formalize fairness in IBM and justify Demographic Parity (DP) as a notion that is particularly well aligned with its semantics.
Yet enforcing DP is computationally challenging:
prior work typically formulates DP as a Linear Programming (LP) problem and relies on costly solvers, rendering them impractical for large-scale networks.
In this paper, we propose a DP-aware objective while maintaining an approximately monotonic submodular structure, enabling efficient optimization with theoretical guarantees.
We integrate this objective with blocking effectiveness through a tunable scalarization, yielding a principled fairness–effectiveness trade-offs.
Building on this structure, we develop CELF-R, an accelerated seed selection algorithm that exploits approximate submodularity to eliminate redundant evaluations and naturally supports Pareto front construction.
Extensive experiments demonstrate that CELF-R consistently outperforms state-of-the-art baselines, achieving a $(1-1/e-\psi)$-approximate solution while maintaining high efficiency.
\end{abstract}

\input{1Introduction}
\input{2Related}

\input{3Preliminary}

\input{4Objective}
\input{5Method}
\input{6Experiment}
\input{7Conclusion}

\appendix

\section*{Ethical Statement}

There are no ethical issues.

\section*{Acknowledgments}
% This work was supported by the Basic Research Program of Jiangsu Province (No. BK20242084) and the National Natural Science Foundation of China (No. 62402496).

% The preparation of these instructions and the \LaTeX{} and Bib\TeX{}
% files that implement them was supported by Schlumberger Palo Alto
% Research, AT\&T Bell Laboratories, and Morgan Kaufmann Publishers.
% Preparation of the Microsoft Word file was supported by IJCAI.  An
% early version of this document was created by Shirley Jowell and Peter
% F. Patel-Schneider.  It was subsequently modified by Jennifer
% Ballentine, Thomas Dean, Bernhard Nebel, Daniel Pagenstecher,
% Kurt Steinkraus, Toby Walsh, Carles Sierra, Marc Pujol-Gonzalez,
% Francisco Cruz-Mencia and Edith Elkind.

\bibliographystyle{named}
\bibliography{ijcai26}

\clearpage
\input{8Appendix}

\end{document}

%% file: math_command.tex
\newtheorem{definition}{Definition}

\newcommand{\R}{\mathbb{R}}

\newcommand{\cC}{\mathcal{C}}

\newcommand{\J}{\mathcal{J}}
\renewcommand{\L}{\mathcal{L}}
\newcommand{\D}{\mathcal{D}}
\newcommand{\M}{\mathcal{M}}

%% file: 1Introduction.tex
\section{Introduction}
% Influence Blocking Maximization (IBM)~\cite{He2012}, a counterpart to Influence Maximization (IM)~\cite{Domingos2001,Kempe2003}, aims to select a positive seed set $S_P$ (budget $k$) to effectively block influence spread from negative seed set $S_N$.
Influence Blocking Maximization (IBM)~\cite{He2012}, a counterpart to Influence Maximization (IM)~\cite{Kempe2003}, aims to select a positive seed set to effectively block influence spread from a negative seed set.
It has a wide range of applications, including rumor suppression~\cite{rumor}, competitive marketing~\cite{marketing}, epidemic control~\cite{epidemic}, {\em etc}.
However, maximizing blocking effectiveness alone can lead to substantial unfairness~\cite{Gu2025,Sun2023}, where some communities receive far less protection than others.

To address this issue, previous works~\cite{Saxena2024} in social network analysis have introduced various notions of fairness, including Max-Min Fairness (MMF)~\cite{Fish2019,Tsang2019}, Welfare Fairness (WF)~\cite{Rahmattalabi2021}, Concave Fairness Framework (CFF)~\cite{Wang2025}, Diversity Constraint (DC)~\cite{Tsang2019}, and Demographic Parity (DP)~\cite{DP}.
However, even in small network, different fairness notions yield markedly different levels of protection across communities.

\begin{figure}[!ht]
  \centering
  \includegraphics[width=0.81\linewidth]{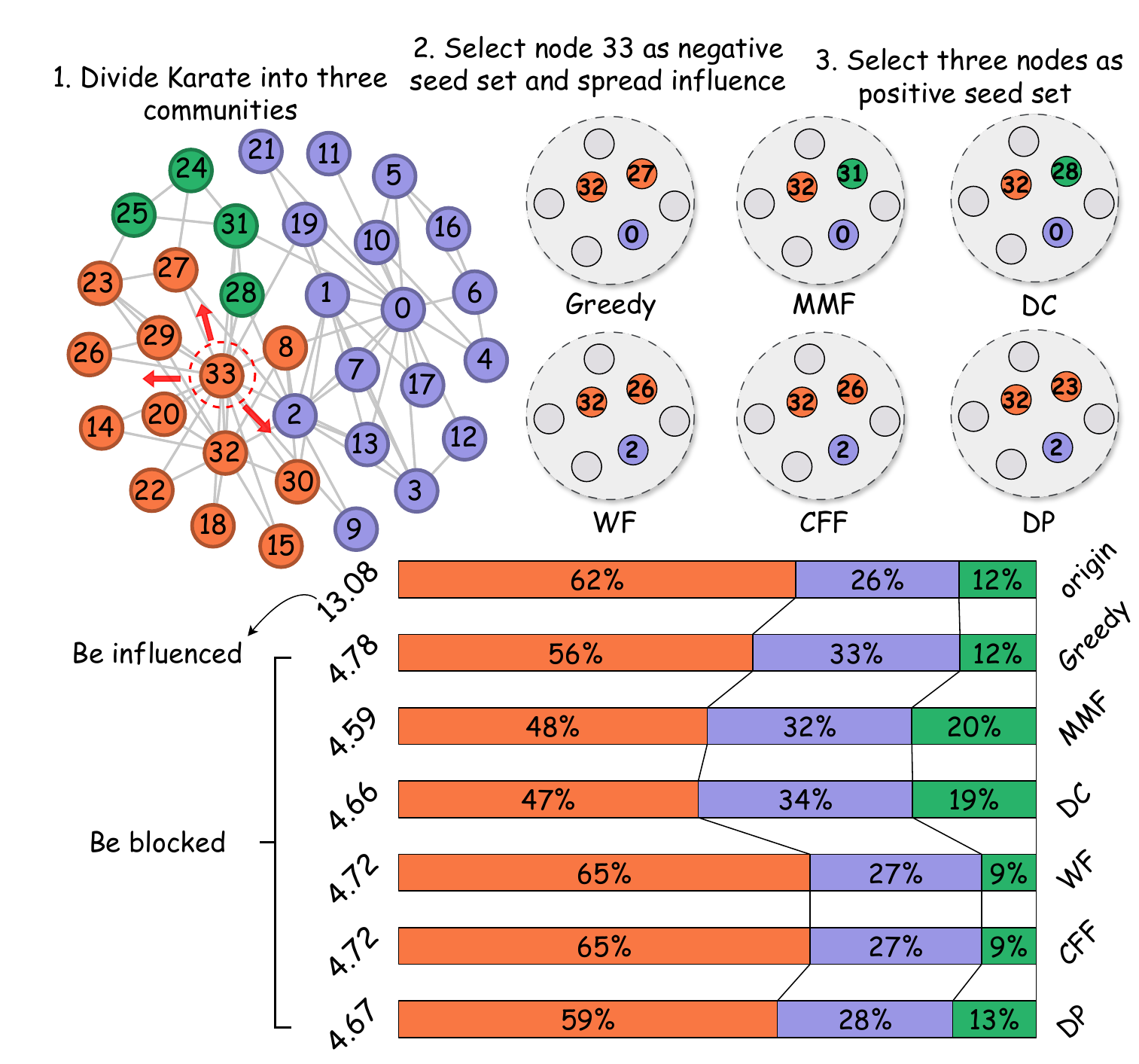}
  \caption{A toy example on Karate Club network.}
  \label{fig:toy example}
\end{figure}

To demonstrate that, we use the Karate Club network (Fig.~\ref{fig:toy example}) as an example.
We select node 33 (the highest-degree node) as the negative seed set, then compute the positive seed set with a budget of 3 according to each fairness notion.
Greedy maximizes the blocking effectiveness without considering fairness.
MMF and DC place one positive seed in green to improve its protection.
However, the final distributions (48\%–32\%–20\% for MMF, 47\%-34\%-19\% for DC) remain imbalanced, with MMF bearing a relatively higher cost as a result.
These shortcomings become even more pronounced as the number of communities increases.
WF and CFF identify the same positive seed set, achieving a reasonably fair allocation.
Yet their results (65\%–27\%–9\%) are still significantly more skewed than the distribution obtained by DP.
They tend to choose well-connected nodes since these can provide relatively good blocking effectiveness and fairness.
However, the less-connected communities ({\em e.g.} the green one) are often overlooked. 
In contrast, DP yields the most balanced distribution (59\%–28\%–13\%), effectively correcting the structural imbalance already present in the original influence spread (62\%–26\%–12\%).

This observation highlights two important principles for fair IBM:
(i) sparsely connected or minority communities should not be disadvantaged; and
(ii) the intervention should not exacerbate existing disparities in negative exposure.
% DP naturally satisfies both, making it a principled fairness criterion for IBM.
DP naturally fits both, making it a principled notion for IBM.

Despite its advantages, DP is difficult to optimize since it lacks an explicit objective function.
Previous work enforces DP via strict constraints, typically formulated and solved as a Linear Programming problem using costly solvers~\cite{Becker2023}.
Such LP-based solutions incur prohibitive computational costs and do not scale to large-scale networks.
This motivates the need for a DP-aware IBM framework that is both theoretically grounded and computationally efficient.

To the best of our knowledge, this is the first work to study Fair Influence Blocking Maximization (FIBM) under DP fairness, providing a scalable solution framework.
Our contributions are summarized as follows:
\begin{itemize}
\setlength{\itemsep}{0pt}
\setlength{\parsep}{0pt}
\item We formalize fairness in IBM and justify why DP is particularly well aligned with its semantics, avoiding the pitfalls observed in other fairness notions.

% \item We introduce an objective that preserves the meaning of DP and prove that it is approximately monotonic submodular, enabling efficient optimization with guarantees.
% Based on this structure, we develop CELF-R, an accelerated selection method that achieves accuracy comparable to full evaluation while eliminating nearly all redundant computations.
\item We introduce a DP-aware objective and prove that it is approximately monotonic submodular, enabling efficient optimization with theoretical guarantees.
It is combined with blocking effectiveness objective through a tunable linear scalarization, allowing flexible control over fairness-effectiveness trade-offs.

% \item Since strict DP may significantly reduce blocking performance, we integrate DP fairness and blocking effectiveness via a tunable linear combination, yielding a flexible spectrum of solutions on the Pareto front.
\item We develop CELF-R, a scalable greedy algorithm tailored to approximately monotonic submodular objectives.
By exploiting bounded deviations in marginal gains, CELF-R achieves high efficiency while retaining solution quality, and naturally supports the construction of Pareto fronts under different trade-off preferences.

% \item Compared with LP-based baselines, our method is substantially more efficient and scalable, as demonstrated by extensive experiments on real-world networks.
% \item Extensive experiments on real-world networks demonstrate that our method  achieves superior trade-offs compared to other baselines while being substantially more efficient and scalable than LP-based methods.
\item Extensive experiments demonstrate that our method consistently outperforms existing baselines, achieving a $(1-1/e-\psi)$-approximate solution while being substantially more scalable than LP-based methods.
\end{itemize}

%% file: 2Related.tex
\section{Related Work}
\subsection{Influence Blocking Maximization}
Influence Maximization (IM)~\cite{Domingos2001} was first formulated as a discrete optimization problem by Kempe {\em et al.}~\cite{Kempe2003}, inspiring a large body of work on scalable algorithms such as RIS-based methods~\cite{IMM}.
In contrast to IM, many real scenarios—rumor containment~\cite{rumor}, competitive marketing~\cite{marketing}, and epidemic control~\cite{epidemic}—require suppressing rather than spreading influence, giving rise to the Influence Blocking Maximization (IBM) problem~\cite{He2012}.

Existing IBM approaches~\cite{ibm1,ibm2,Gu2025} improve blocking effectiveness using heuristics, sampling, or fast influence estimators.
However, these methods implicitly assume that all communities benefit equally from intervention, neglecting disparities in how different groups face and resist negative influence—an issue central to fair decision-making.

\subsection{Fairness Notions}
Previous works~\cite{Saxena2024,Dong2023} in social networks have introduced various fairness notions.
Max-Min fairness (MMF)~\cite{Fish2019,Tsang2019} protects the worst-off group but often sacrifices substantial overall utility.
Diversity constraints (DC)~\cite{Tsang2019} require group-level improvement relative to proportional baselines, yet depend heavily on internal connectivity and break down in sparse or imbalanced communities.
Welfare-based notions ({\em e.g.}, WF~\cite{Rahmattalabi2021}, CFF~\cite{Wang2025}) aggregate utilities through concave transformations with adjustable inequality aversion.
However, Rahmattalabi {\em et al.}~\cite{Rahmattalabi2021} point out that it cannot guarantee utility gap reduction under general network structures with inter-community connections.
Moreover, it tends to favor structurally advantaged groups, thus overlooking sparsely connected or minority communities.

Demographic Parity (DP)~\cite{Stoica2019} provides a structure-agnostic criterion by directly bounding inter-group utility gaps.
It neither relies on community topology nor exacerbates pre-existing disparities, making it particularly suitable for IBM, where negative exposure differs widely across communities.
Yet enforcing DP remains challenging: prior work typically encodes DP via linear programming constraints~\cite{Becker2023}, which scale poorly and lack compatibility with greedy or submodular-based optimization methods essential for large networks.

%% file: 3Preliminary.tex
\section{Preliminaries}
We model a social network as a directed graph $G=(V,E)$, where each node $v\in V$ belongs to one community in a disjoint partition $\mathcal{C}=\{c_1,\dots,c_{|\mathcal{C}|}\}$.
Each edge $(u,v)$ carries an influence weight $p_{u,v}$.
% We adopt the Linear Threshold (LT) model~\cite{Kempe2003}, a standard diffusion model in influence blocking studies~\cite{Sun2023,He2012}.
In this paper, we select the Linear Threshold (LT)~\cite{Kempe2003} model as the diffusion model, which is widely used in literature~\cite{Sun2023,LT1,He2012}.
Let $\sigma(S,G)$ denote the expected influence spread of seed set $S$, and $\sigma_c(S,G)$ the spread restricted to community $c$.

\subsection{Fair Influence Blocking Maximization}
Let $S_N$ be the negative seed set. A positive seed set $S_P$ acts purely as immunized nodes: they neither receive nor transmit influence.
The blocking effectiveness is
\begin{equation}
\sigma^-(S_P)=\sigma(S_N,G)-\sigma(S_N,G\setminus S_P),
\end{equation}
and $\sigma_c^-(S_P)$ is defined analogously for community $c$.

Classical IBM seeks:
\begin{equation}
S^\#_P=\mathop{\arg\max}\limits_{S_P \subseteq V \backslash S_N, |S_P| \leq k} \sigma^-(S_P).
\end{equation}

To enforce fairness, we integrate DP by requiring that communities achieve similar blocked ratios,
\begin{equation}
\frac{\sigma_c^-(S_P)}{\sigma_c(S_N,G)}.
\end{equation}

This yields the following formulation.
\begin{definition}[Fair Influence Blocking Maximization, FIBM] \label{def:FIBM}
Given a graph $G=(V,E)$, community partition $\mathcal{C}=\{c_1,\dots,c_{|\mathcal{C}|}\}$, budget $k$, and negative seed set $S_N$, FIBM seeks a positive seed set $S_P$ such that
\begin{equation}
\begin{aligned}
&S^*_P = \mathop{\arg\min}\limits_{S_P \in V\backslash S_N, |S_P|\leq k} ( \max\limits_{c\in\cC} \frac{\sigma_c^-(S_P)}{\sigma_c(S_N, G)}-\min\limits_{c\in\cC} \frac{\sigma_c^-(S_P)}{\sigma_c(S_N, G)} ) \\
&\text{s.t.} \quad 1 - \frac{\sigma^-(S_P^*)}{\sigma^-(S_P^\#)} \leq \mu \\
\end{aligned}
\end{equation}
where $\mu$ controls the allowable loss in blocking effectiveness.
\end{definition}

This formulation captures the inherent trade-offs between fairness and blocking effectiveness, since achieving DP should not excessively degrade overall blocking performance.

% \subsection{Approximate Submodularity and Approximate Monotonicity}
% \subsection{Approximately Monotonic Submodular Optimization}
\subsection{Approximate Structure}
Classical greedy methods rely on objectives that are monotonic submodular, ensuring diminishing returns and tractable optimization.
However, fairness-aware objectives rarely satisfy these properties exactly.
Prior studies~\cite{Tsang2019} show that MMF and DC are non-submodular, and similar violations arise in DP.

In practice, however, the fairness-aware objective behaves almost monotonic and submodular throughout most of the optimization process: marginal gains decrease predictably, and deviations mainly occur near highly balanced (fair) solutions.
This motivates adopting a relaxed structure that tolerates small discrepancies while retaining the algorithmic advantages of submodularity.

Therefore, we adopt the framework of approximate submodularity~\cite{Krause2010,Qian2019,Singer2018} and approximate monotonicity.
Formally, these approximate properties are defined as
\begin{definition}[$(\kappa, \epsilon)$-approximately monotonic submodular] \label{def:approximately monotonic submodular}
A set function $f: 2^V \to \R$ is called $(\kappa, \epsilon)$-approximately monotonic submodular (with parameters $\kappa \geq 0$ and $\epsilon \geq 0$) if it satisfies the following two conditions simultaneously:
\begin{itemize}
\setlength{\itemsep}{0pt}
\setlength{\parsep}{0pt}
   \item Approximate Monotonicity:$$f(X \cup \{v\}) \ge f(X) - \kappa$$
   
   \item Approximate Submodularity:$$f(X \cup \{v\}) - f(X) \ge f(Y \cup \{v\}) - f(Y) - \epsilon$$
\end{itemize}
for all $X\subseteq Y\subseteq V$ and $v\notin Y$.
\end{definition}
Here, $\kappa$ and $\epsilon$ measure deviations from exact structure and vanish when $f$ is truly monotone submodular.
This relaxed model enables the use of efficient greedy-style algorithms such as CELF-R developed later.

%% file: 4Objective.tex
\section{Objective Function}
In the FIBM problem, the positive seed set $S_P$ must simultaneously (i) suppress negative influence and (ii) distribute protection fairly across communities under Demographic Parity (DP). We therefore construct an objective that captures both fairness and blocking effectiveness while supporting efficient approximate-submodular optimization.

\subsection{Fairness-Aware Objective}
\paragraph{DP fairness and its surrogate formulation.}
Under DP, the share of influence blocked in community $c$ should match its share of the negative exposure:
\begin{equation}
n_c=\frac{\sigma_c(S_N,G)}{\sigma(S_N,G)},\qquad
x_c(S_P)=\frac{\sigma^-_c(S_P)}{\sigma^-(S_P)} .
\end{equation}

Directly enforcing $x_c = n_c$ is combinatorially difficult.
We instead adopt a concave surrogate, a common approach for modeling inequality aversion:
\begin{equation}\label{eq:W}
\mathcal{W}(S_P)=\sum_{c\in\mathcal{C}} r_c (x_c(S_P))^\alpha,\qquad 0<\alpha<1 ,
\end{equation}
with $r_c=n_c^{1-\alpha}$.
Concavity amplifies improvements in under-protected communities, encouraging the blocked proportions $x_c$ to align with the target proportions $n_c$.

\begin{restatable}{lemma}{lemmaWDP}
\label{lemma: W to DP}
Maximizing $\mathcal{W}(\cdot)$ increases proportional alignment between $x_c$ and $n_c$; $\mathcal{W}(\cdot)=1$ corresponds to strict Demographic Parity.
\end{restatable}

\begin{restatable}{lemma}{lemmaAlpha}
\label{lemma: alpha analysis}
when $|\frac{x_c}{n_c}-1|<\phi$, for any $\alpha_1, \alpha_2 \in (0,1)$, $|\mathcal{W}_{\alpha_1}-\mathcal{W}_{\alpha_2}| \leq \frac{|(\alpha_1-\alpha_2)(\alpha_1+\alpha_2-1)|}{2}\phi^2+O(\phi^3)$.
\end{restatable}

\paragraph{Approximate structure of $\mathcal{W}(\cdot)$}
Although $\mathcal{W}(\cdot)$ promotes DP fairness, it is generally neither monotonic nor submodular.
Adding a node may rebalance the proportions $x_c$, leading to small decreases in $\mathcal{W}(\cdot)$ (non-monotonicity), and marginal gains may occasionally increase (submodularity violations).
These deviations tend to occur when the solution is already approaching proportional balance; elsewhere, $\mathcal{W}(\cdot)$ behaves closely to a monotonic submodular function.

\begin{restatable}{lemma}{lemmaW}
\label{Lemma: W is approximately monotonic submodular}
$\mathcal{W}(\cdot)$ is $(\kappa,\epsilon)$-approximately monotonic submodular.
\end{restatable}

Detailed proofs for Lemmas are provided in the Appendix A.1.

\subsection{Blocking Effectiveness Objective}

To quantify the suppression of negative influence, we use the normalized blocking measure
\begin{equation}\label{eq:F}
\mathcal{F}(S_P)=\frac{\sigma^-(S_P)}{\sigma(S_N,G)}.
\end{equation}

Normalization places $\mathcal{F}(\cdot)$ and $\mathcal{W}(\cdot)$ on comparable scales.

We additionally enforce a tolerance level $\mu$ on effectiveness loss:
\begin{equation}
1 - \frac{\mathcal{F}(S_P)}{\mathcal{F}(S_P^\#)} \le \mu,
\end{equation}
where $S_P^\#$ maximizes blocking alone.
Previous work~\cite{Sun2023} shows that $\mathcal{F}(\cdot)$ is monotone and submodular.

\subsection{Final Combined Objective}
To jointly optimize fairness and effectiveness, we apply linear scalarization:
\begin{equation} \label{eq:K}
\mathcal{K}(S_P)=\beta\cdot\mathcal{W}(S_P)+(1-\beta)\cdot\mathcal{F}(S_P),
\qquad 0\le\beta\le1 .
\end{equation}

The weight $\beta$ controls the fairness–effectiveness trade-offs and yields a Pareto spectrum of solutions.
Because both terms are approximate submodular, $\mathcal{K}(\cdot)$ inherits the same structural properties.

\begin{figure*}[!ht]
    \centering
    \includegraphics[width=1\linewidth]{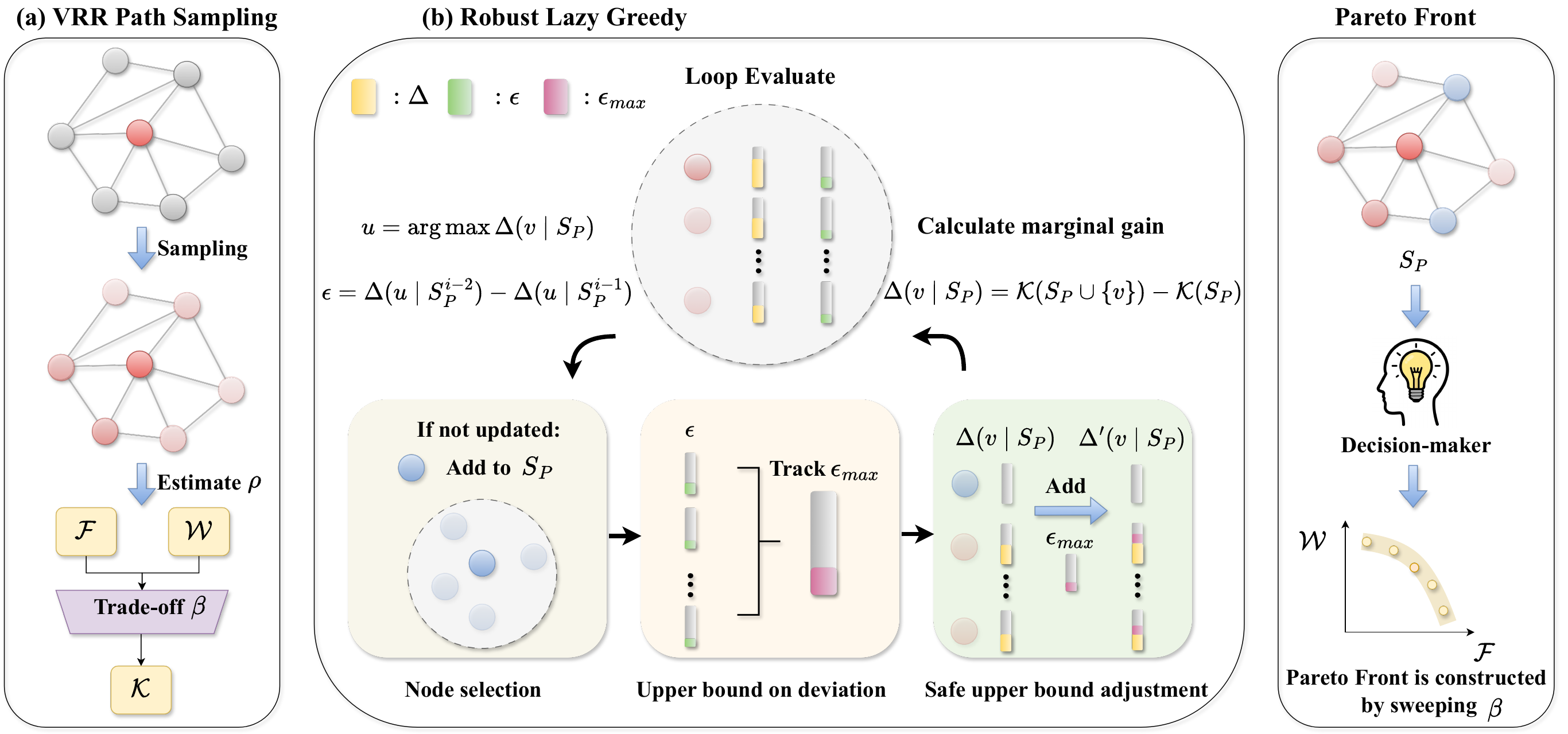}
    \caption{Overview of CELF-R and Pareto Front Construction. CELF-R consists of two main components: (a) Naive VRR path sampling, which estimates node-level blocking contributions and supports efficient evaluation of the trade-off objective $\mathcal{K}$; and (b) a robust lazy greedy selection strategy that explicitly models bounded deviations under approximate submodularity. By sweeping the trade-off parameter $\beta$, CELF-R produces a set of solutions that collectively form an empirical Pareto front balancing fairness and blocking effectiveness.}
    \label{fig:CELF-R}
\end{figure*}

\begin{restatable}{lemma}{lemmaopt}
\label{lemma:opt}
    For our proposed objective function $\mathcal{K}(\cdot)$ and any $\chi \in \{0, 1\}^n$, there exists an element $v \notin \chi$ that satisfies
    \begin{equation*}
        \mathcal{K}(\chi \cup \{v^*\}) - \mathcal{K}(\chi) \geq \frac{1}{k}(OPT - \mathcal{K}(\chi)) - \kappa - \epsilon,
    \end{equation*}
    where $k$ is the budget for the positive seed set.
\end{restatable}
% \begin{restatable}{lemma}{lemmaopt}
% \label{lemma:opt}
% For $\mathcal{K}(\cdot)$ and any partial solution $\chi$, there exists $v\notin\chi$ such that
% \begin{equation}
% \mathcal{K}(\chi\cup{v})-\mathcal{K}(\chi)
% \ge \frac{1}{k}\bigl(OPT-\mathcal{K}(\chi)\bigr)-\kappa-\epsilon.
% \end{equation}
% \end{restatable}

\begin{restatable}{theorem}{theoremK}
\label{therorem: K}
    For a $(\kappa, \epsilon)$-approximately monotonic submodular function $\mathcal{K}(\cdot)$ with a size constraint $k$, let $\chi^*$ be the optimal set that maximizes $\mathcal{K}(\cdot)$. 
    % Then $\mathcal{K}(\chi) \geq (1 - 1/e) \cdot (\mathcal{K}(\chi^*) - k\kappa - k\epsilon)$ with $\chi$ selected with greedy approach. 
    Then $\mathcal{K}(\chi) \geq (1 - 1/e) \cdot \mathcal{K}(\chi^*)-\psi$ with $\chi$ selected with greedy approach. 
\end{restatable}
% \begin{restatable}{theorem}{theoremK}
% \label{therorem: K}
% For a $(\kappa,\epsilon)$-approximately monotonic submodular objective under a cardinality constraint $k$, the greedy algorithm returns $\chi$ satisfying
% \begin{equation}
% \mathcal{K}(\chi)\ge (1-1/e)\bigl(\mathcal{K}(\chi^*)-k\kappa-k\epsilon\bigr).
% \end{equation}
% \end{restatable}

Please refer to the Appendix A.1 for their detailed proofs.

%% file: 5Method.tex
\section{Method}
In this section, we present our proposed CELF-R algorithm.
As illustrated in Fig.~\ref{fig:CELF-R}, CELF-R consists of two components:
(a) Naive Valid Reverse-Reachable (VRR) path sampling,
and (b) a robust lazy greedy selection strategy.
Moreover, we provide the construction of Pareto front for fairness-effectiveness trade-offs via CELF-R.

\subsection{Naive VRR Path Sampling} 
We adopt the Naive VRR path sampling approach proposed by Sun {\em et al.}~\cite{Sun2023}, which estimates the number of VRR paths $\theta$ required to guarantee approximation quality with high probability.
Accordingly, we use the same value of $\theta$ throughout this work.

Let $\theta_v$ denote the number of paths rooted at $v$, $\M[u][v]$ denote the number of occurrences of node $u$ in paths originating from $v$, and $\D[u]$ be the set of VRR paths containing $u$.
The frequency of each VRR path $P_v$ is recorded in $\L[P_v]$.
We maintain a tuple $\J=(\M,\D,\L)$.

The blocking capability of a node $u$ is then represented as
\begin{equation}
\rho_u = \{(v, \M[u][v]/\theta_v) \mid v \in V \setminus S_N\}.
\end{equation}

When a node $u$ is selected, all VRR paths in $\D[u]$ are invalidated.
Specifically, for each path $P_v \in \D[u]$ and each node $w \in P_v$, we update $\M[w][v] = \M[w][v] - \L[P_v]$ and set $\L[P_v]=0$.
Therefore, $\mathcal{K}(\cdot)$ can be easily calculated via $\rho$.

\subsection{Lazy Greedy Optimization under Approximate Submodularity}
In the previous section, we established that the combined objective $\mathcal{K}(\cdot)$ is
$(\kappa,\epsilon)$-approximately monotonic submodular.
This implies that classical greedy optimization remains effective, but strict diminishing returns
and monotonicity may be locally violated.
As a result, directly applying the Cost-Effective Lazy Forward (CELF) algorithm~\cite{celf},
which relies on exact submodularity, is no longer sound.

To address this issue, we propose CELF-R, an enhanced variant of CELF tailored for approximately submodular objectives.
Unlike classical CELF, which assumes a fixed upper bound on marginal gains, CELF-R explicitly models bounded deviations.
The deviation parameter $\epsilon$ is treated as an \emph{empirical upper bound} on marginal gain
violations, estimated online during the greedy process rather than assumed as a known global constant.

Algorithm~\ref{alg:CELF-R} details the procedure.
At each iteration, CELF-R maintains a lazy evaluation of marginal gains $\Delta(v \mid S_P)$ together with an \texttt{updated} flag indicating whether the gain has been recomputed under the current solution.
When a candidate node $u$ is selected, its marginal gain is recomputed exactly, and the deviation $\epsilon$ is estimated based on the marginal decay of $u$ across consecutive stages.
This quantity captures the maximum observed violation of diminishing returns and is accumulated into $\epsilon_{\max}$.

Under $(\kappa,\epsilon)$-approximate submodularity, the true marginal gain of any element can increase by at most $\epsilon$.
Therefore, compensating stale gains by $\epsilon_{\max}$ preserves a valid upper bound on their true marginal gains.
This allows CELF-R to safely skip redundant recomputations while maintaining the correctness of greedy selection.

In practice, violations of monotonicity captured by $\kappa$ are rarely observed and tend to occur only when the solution is near saturation, {\em i.e.}, when additional selections provide limited improvement.
Empirically, the magnitude of such violations is dominated by $\epsilon$, and their effect is effectively absorbed by the same upper-bound adjustment mechanism.
We therefore focus on $\epsilon$ in implementation without loss of practical effectiveness.

\begin{algorithm}[!b]
\caption{CELF-R}
\label{alg:CELF-R}
\begin{algorithmic}[1]
\Require 
Graph $G=(V,E)$, community $\mathcal{C}$, negative seed set $S_N$, budget $k$, trade-off parameter $\beta$.
\Ensure 
Positive seed set $S_P$.

\State $\J(\M, \D, \L)$ = Naive-VRRP($G, \mathcal{C}, S_N$)
\State Initialize $S_P = \emptyset$, $\epsilon_{\max} = 0$
\State Initialize $\Delta(v \mid S_P) = \mathcal{K}(S_P \cup \{v\}) - \mathcal{K}(S_P)$ for all $v \in V \setminus S_N$

\For{$i = 1$ to $k$}
    \State $updated[v] = false$ for all $v \in V \setminus (S_N \cup S_P)$
    \While{true}
        \State $u = \arg\max_{v \in V \setminus (S_N \cup S_P)} \Delta(v \mid S_P)$
        \If{$updated[u] = true$}
            \State \textbf{break}
        \EndIf
        \State Recompute $\Delta(u \mid S_P) = \mathcal{K}(S_P \cup \{u\}) - \mathcal{K}(S_P)$
        \If{$i > 2$}
            \State $\epsilon = \Delta(u \mid S_P^{i-2}) - \Delta(u \mid S_P^{i-1})$ //superscript denotes the i-th state.
            \State $\epsilon_{\max} = \max(\epsilon_{\max}, \epsilon)$
        \EndIf
        \State $updated[u] = true$
    \EndWhile

    \For{each $v \in V \setminus (S_N \cup S_P)$}
        \If{$updated[v] = false$}
            \State $\Delta(v \mid S_P) = \Delta(v \mid S_P) + \epsilon_{\max}$
        \EndIf
    \EndFor

    \State $S_P = S_P \cup \{u\}$, update $\J$
\EndFor
\State \Return $S_P$
\end{algorithmic}
\end{algorithm}

\paragraph{Implementation and Efficiency.}
For clarity, Algorithm~\ref{alg:CELF-R} presents the core sequential logic of CELF-R.
In practice, several components—such as VRR path sampling, marginal gain recomputation, and gain compensation—can be efficiently parallelized or executed in batches.
Our experimental implementation adopts multi-threading and batched updates, which significantly
improves runtime efficiency without affecting the algorithmic behavior or solution quality.

\subsection{Pareto Front Construction via CELF-R}
For a fixed $\beta$, CELF-R returns a single solution optimizing the corresponding scalarized objective.
By varying $\beta \in [0,1]$, we obtain a collection of trade-off solutions that together form an empirical Pareto front using Fast Non-dominated Sorting.

In practice, domain-specific constraints—such as the allowable loss in blocking effectiveness $\mu$—can be used to restrict the feasible range of $\beta$.
Within this range, finer-grained values of $\beta$ may be explored to better capture application-specific preferences.
In most cases, the solution tends to favor fairness when $\beta \rightarrow 1$, otherwise the opposite.
We deliberately refrain from enforcing a fixed selection criterion, as the final choice along the Pareto front depends on external requirements and deployment contexts.

Importantly, the Pareto front construction incurs only marginal overhead.
The VRR paths samples and auxiliary data structures are shared across runs with different $\beta$, and CELF-R naturally supports warm-starting.
As a result, the full trade-off curve can be obtained efficiently, making the approach practical for large-scale networks.

%% file: 6Experiment.tex
\section{Experiments}
\subsection{Experimental Setups}
\paragraph{Datasets.}
We conduct experiments on four real-world social networks.
Facebook, Slashdot, and Gowalla are obtained from~\cite{nr}, where communities are detected using the Louvain algorithm~\cite{Louvain}.
For Pokec~\cite{pokec}, communities are divided by user regions.
Details are presented in Table~\ref{Tab:Datasets}.

\begin{table}[!ht]
    \centering
    \caption{Datasets}
    \label{Tab:Datasets}
    \begin{tabular}{cccccc}
    \toprule
    \textbf{Dataset} & \textbf{$|V|$} & \textbf{$|E|$} & \textbf{$|C|$} & \textbf{$\left\langle{d}\right\rangle$} & \textbf{Type}\\
    \midrule
     Facebook &$4.5$K &$161.4$K &$9$   &$141.5$ &Undirected \\
     Slashdot &$70.0$K &$358.6$K &$358$   &$20.5$ &Undirected\\
     Gowalla &$196.6$K &$950.3$K &$518$   &$19.3$ &Undirected\\
     Pokec    &$1.6$M  &$30.6$M  &$10$    &$37.5$ &Directed \\
    \bottomrule
    \end{tabular}
\end{table}

\paragraph{Setups.}
Following the standard adversarial assumption for IBM problems, we select the top-degree nodes as the negative seed set, setting its size to 50, and define the propagation probability as $p = 1/d^-(v)$.
We set the fairness parameter $\alpha=0.5$; as shown in Lemma~\ref{lemma: alpha analysis}, varying $\alpha$ has limited impact on the overall results.
All experiments are implemented in Golang and conducted on a server equipped with an AMD EPYC 7402P 24-Core processor and 64GB RAM.
We run each experiment 5 times and take the average.

\paragraph{Performance Metrics.}
We evaluate solutions using the fairness-aware objective $\mathcal{W}(\cdot)$ and the blocking effectiveness objective $\mathcal{F}(\cdot)$, defined in Eq.~(\ref{eq:W}) and Eq.~(\ref{eq:F}), respectively.
Both metrics are maximization objectives, where higher values indicate better performance.

\begin{figure*}[!t]
  \centering
  \subfigure[Facebook]
  {\includegraphics[width=0.245\linewidth]{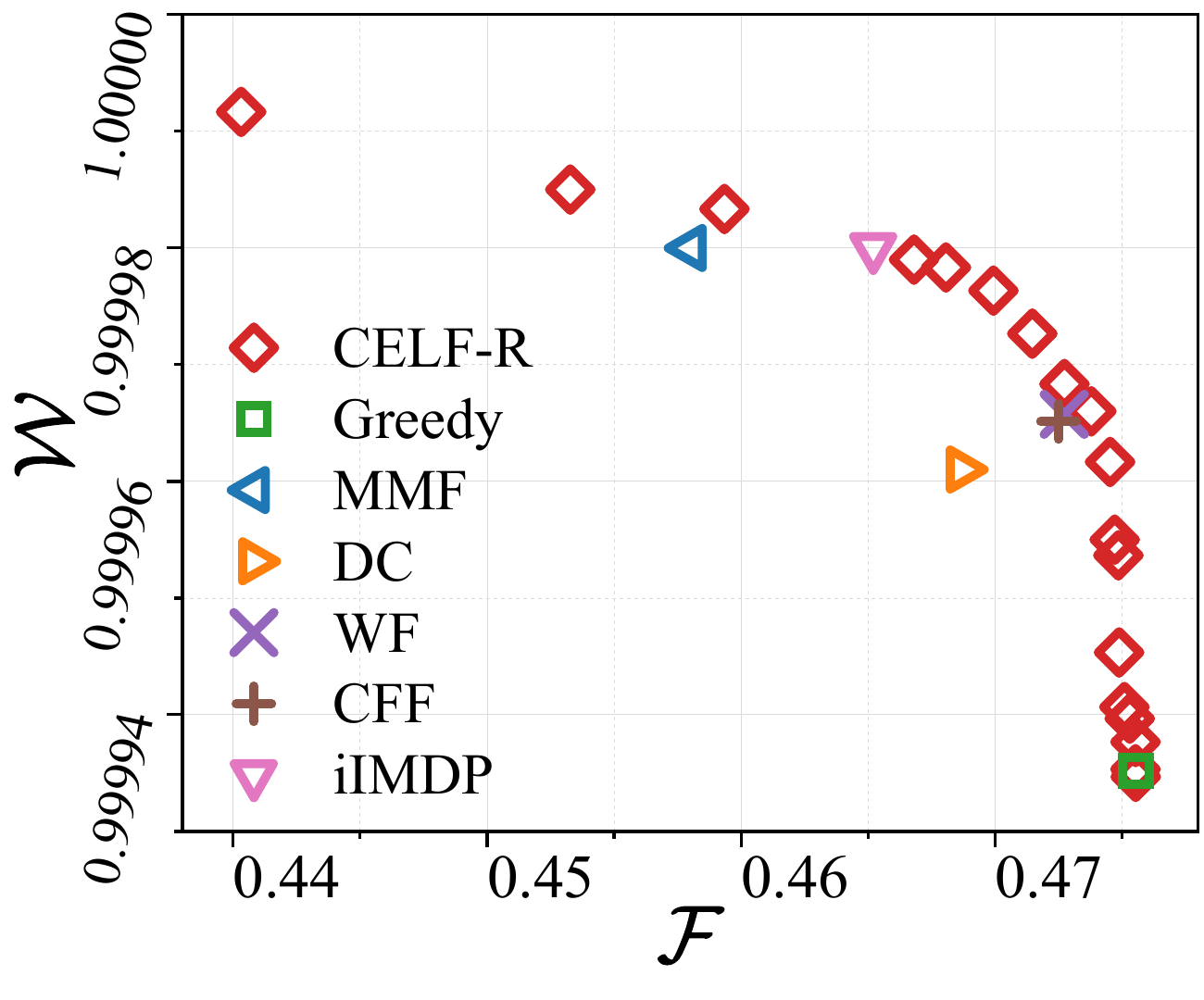}}
  \subfigure[Slashdot]
  {\includegraphics[width=0.245\linewidth]{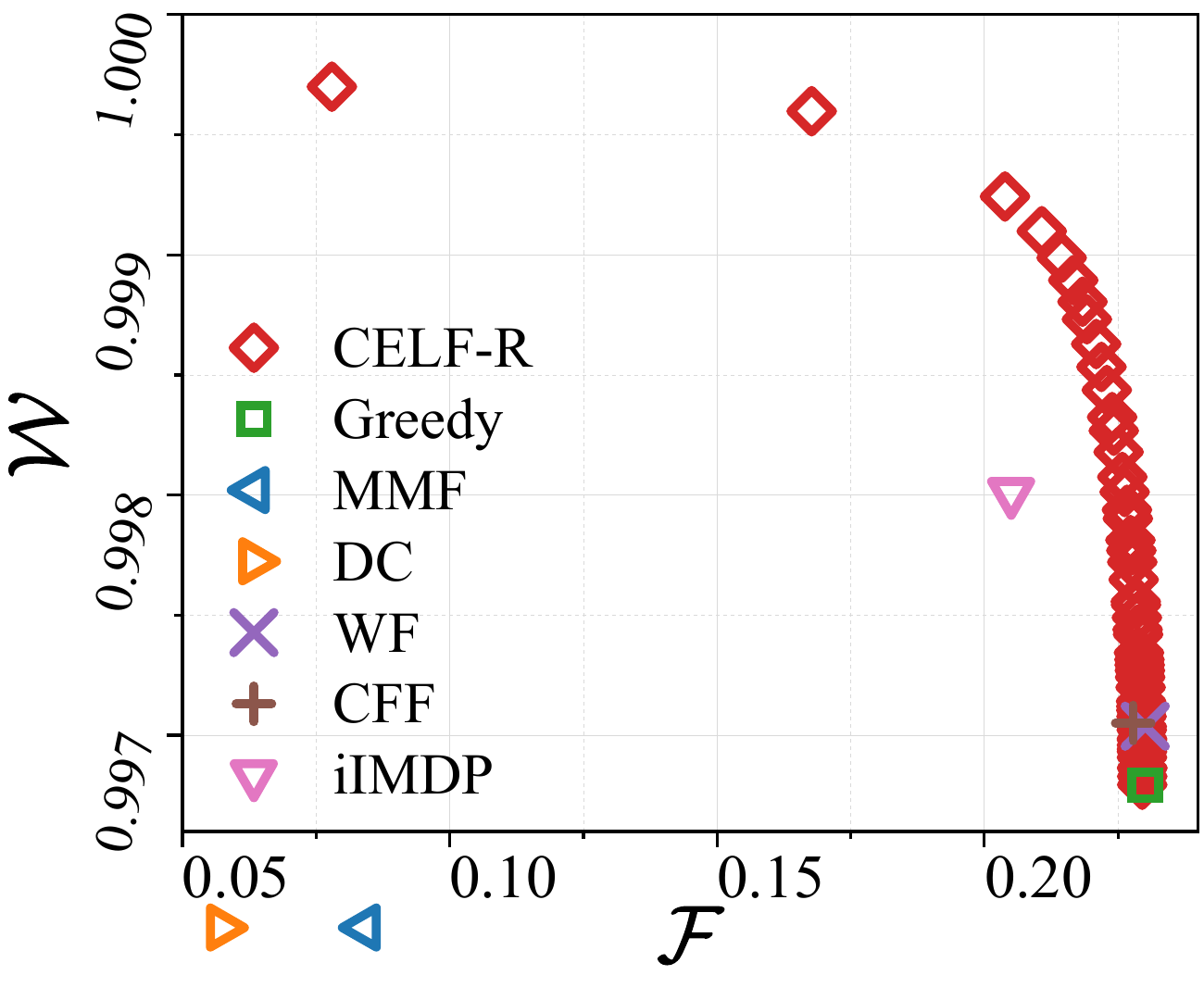}}
  \subfigure[Gowalla]
  {\includegraphics[width=0.245\linewidth]{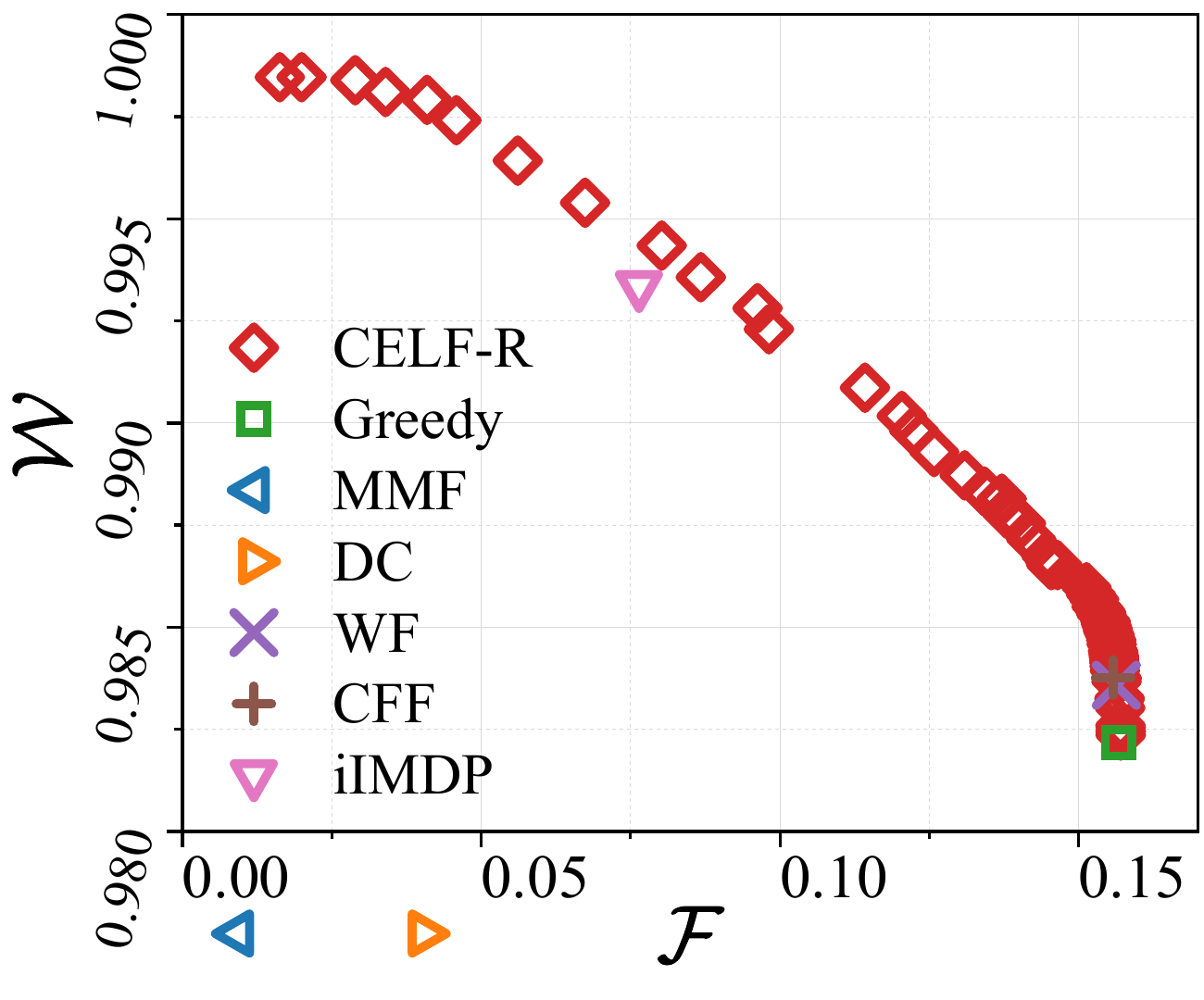}}
  \subfigure[Pokec]
  {\includegraphics[width=0.245\linewidth]{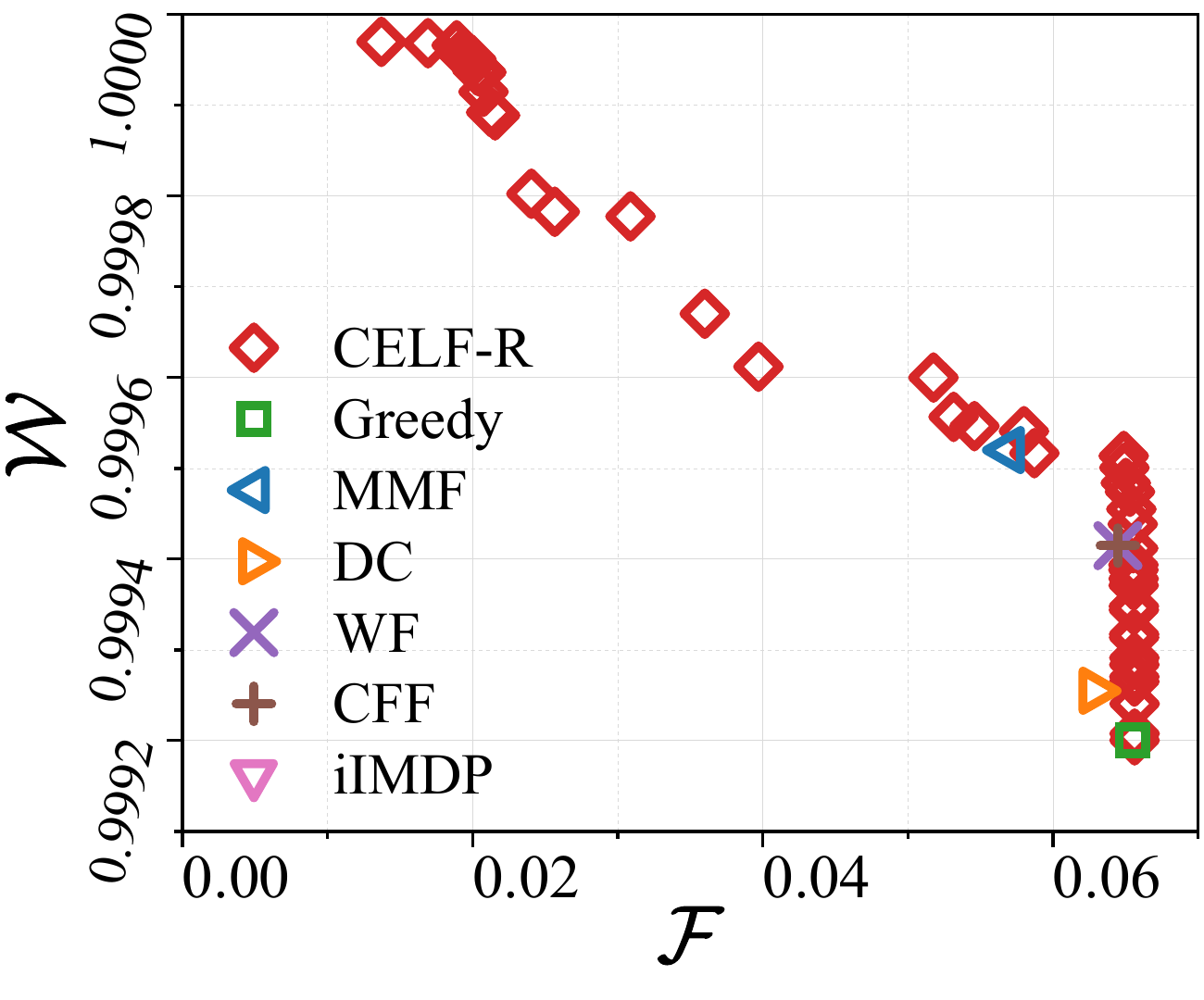}}
  % \caption{Comparison of different fairness notions.}
  \caption{Comparison of different fairness notions on the FIBM problem.
The x-axis denotes blocking effectiveness $\mathcal{F}$, and the y-axis denotes the fairness-aware objective $\mathcal{W}$, where higher values are better for both.
Our method produces a Pareto front by varying the trade-off parameter $\beta$, while each baseline yields a single solution.
Points closer to the top-right corner indicate better fairness–effectiveness trade-offs.}
  \label{fig:Comparisons}
\end{figure*}

\begin{figure*}[ht]
  \centering
  \subfigure[Facebook]
  {\includegraphics[width=0.245\linewidth]{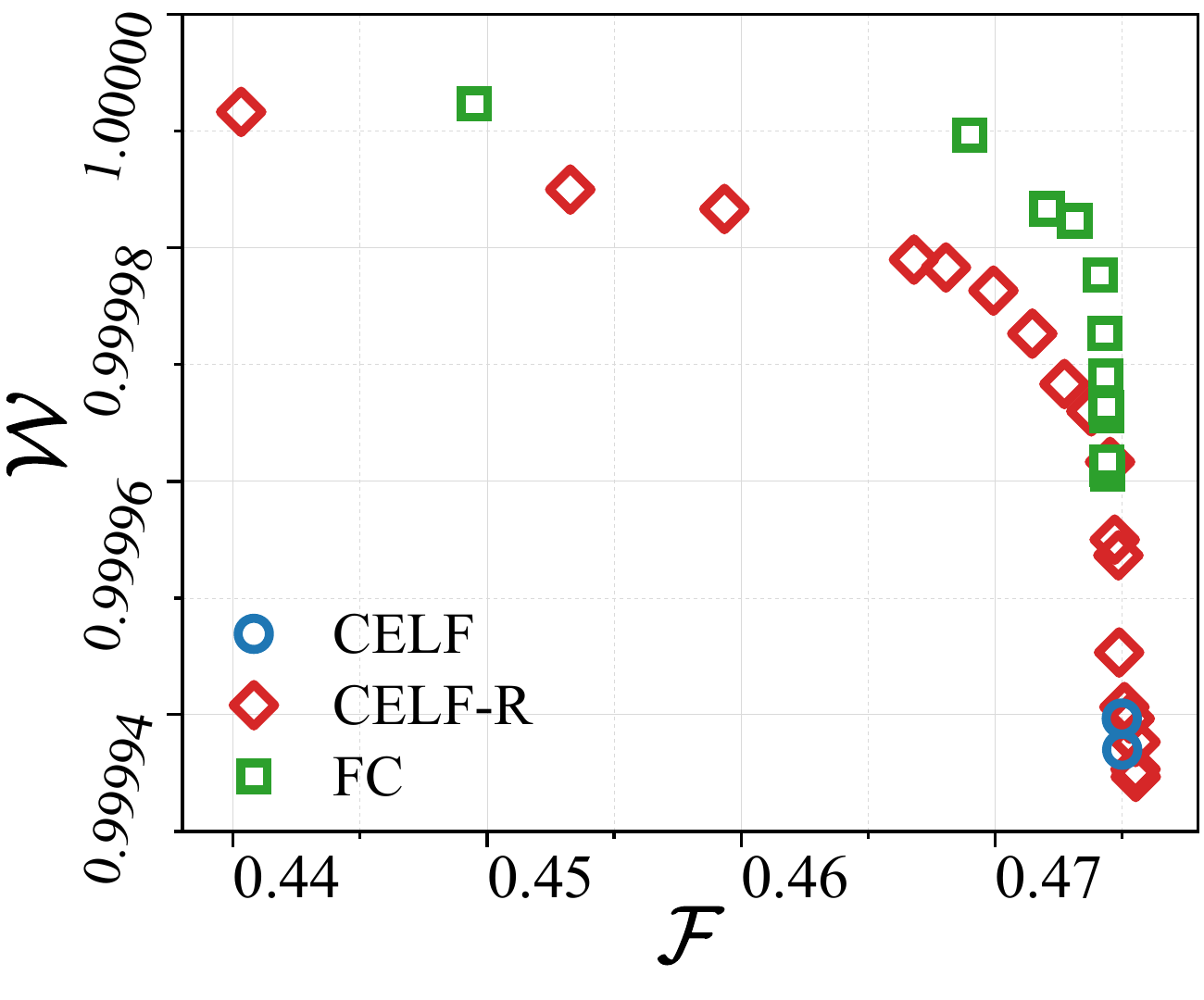}}
  \subfigure[Slashdot]
  {\includegraphics[width=0.245\linewidth]{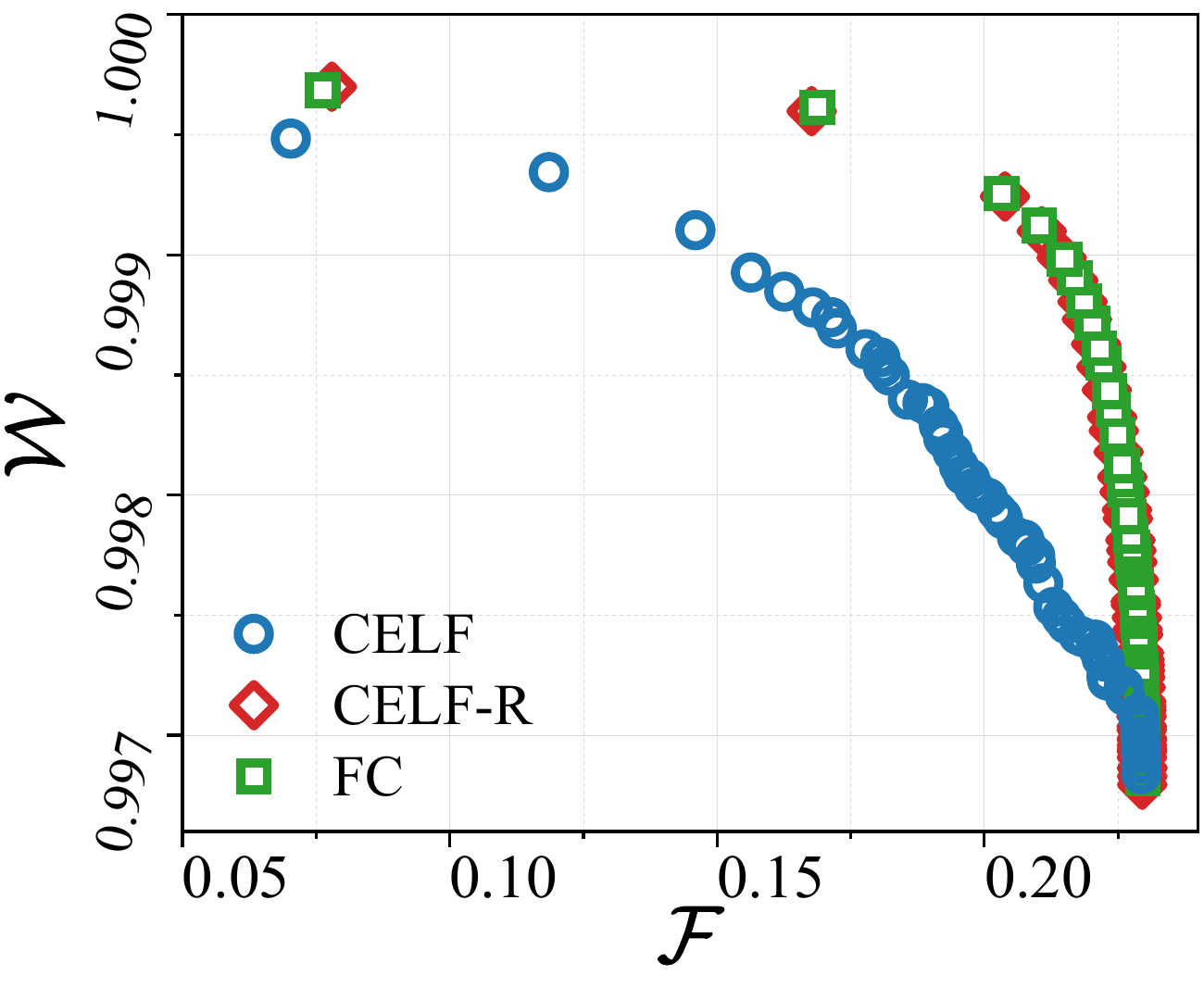}}
  \subfigure[Gowalla]
  {\includegraphics[width=0.245\linewidth]{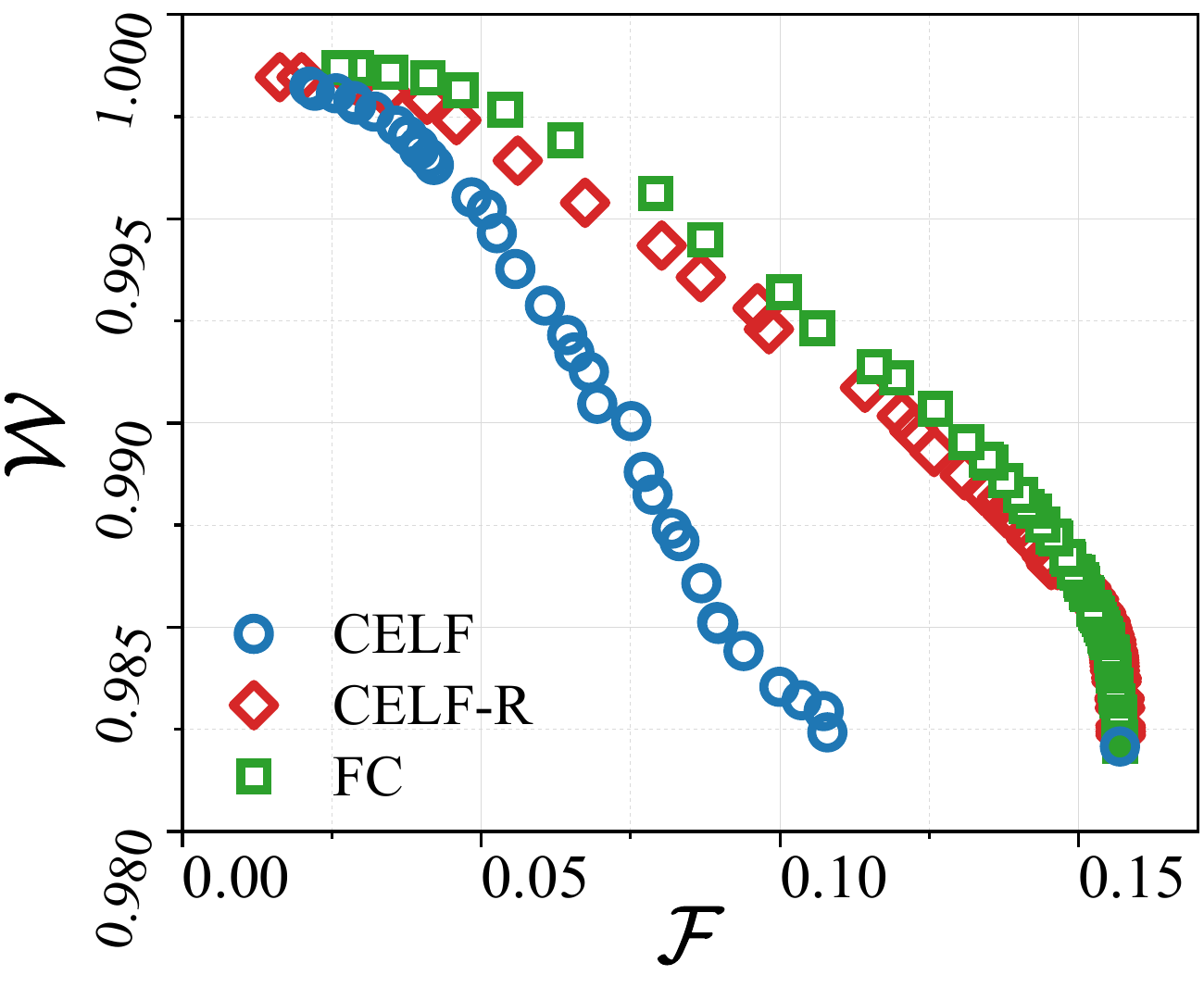}}
  \subfigure[Pokec]
  {\includegraphics[width=0.245\linewidth]{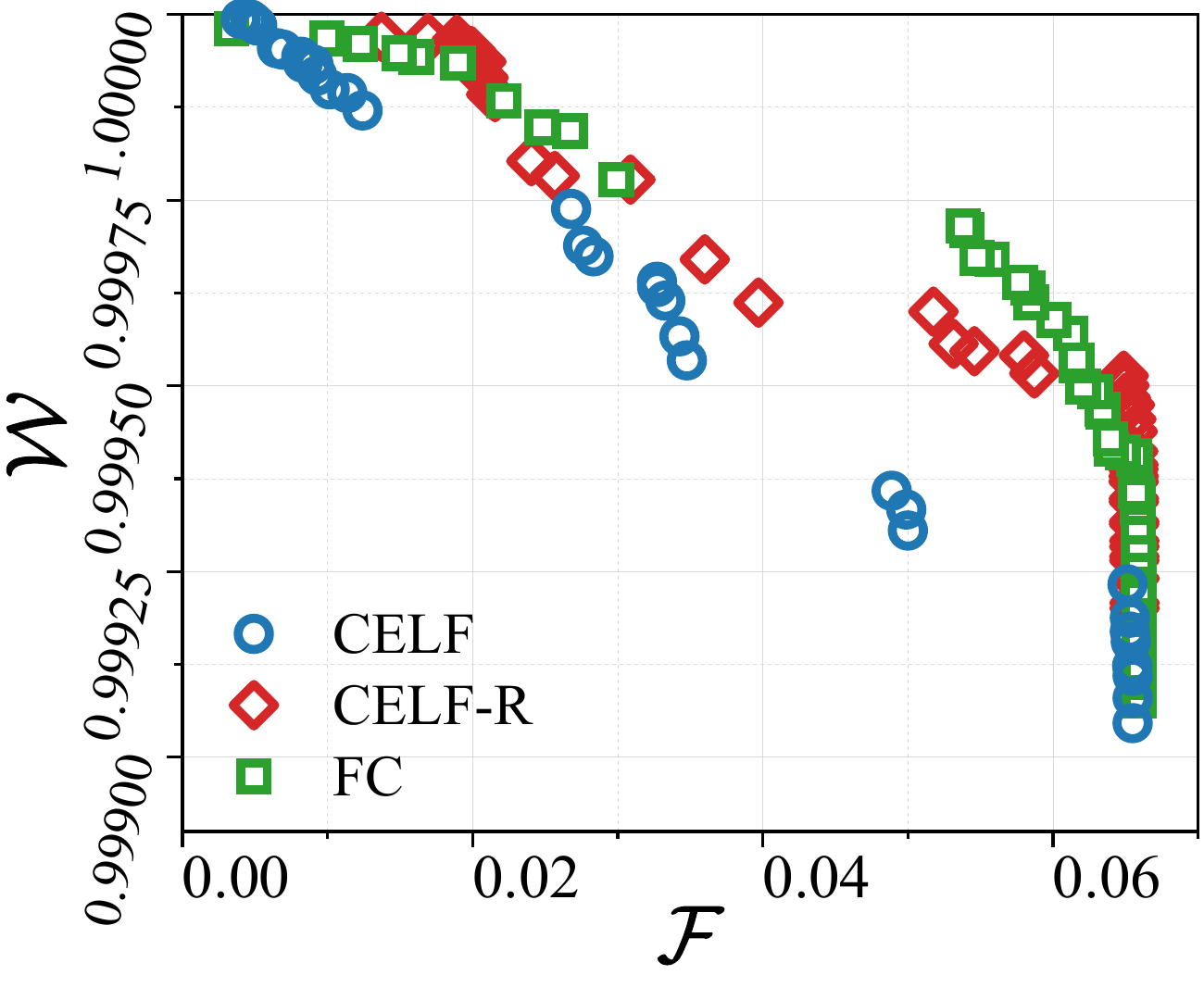}}
  % \caption{CELF, CELF-R, and FC.}  
  \caption{Ablation study on seed selection strategies.
Pareto fronts produced by CELF-R and Full Computation (FC) nearly overlap, indicating that CELF-R achieves solution quality comparable to exhaustive marginal gain evaluation.
In contrast, CELF performs significantly worse due to its assumption of strict submodularity.}
  \label{fig:ablation_fs}
\end{figure*} 

\subsection{Comparison with Baselines}
In this subsection, we construct an empirical Pareto front by varying the trade-off parameter $\beta$ from $0$ to $1$ with a step size of $0.01$.
This full-sweeping strategy is adopted solely for demonstration purposes, in order to obtain a relatively comprehensive view of the fairness–effectiveness trade-offs.
% In practical deployments, many of these solutions may be infeasible due to application-specific constraints such as the allowable loss in blocking effectiveness $\mu$.
In fact, many of these solutions may be infeasible due to constraints like the allowable loss in blocking effectiveness $\mu$.

We set the budget $k=100$ for all networks and compare the resulting positive seed set with the following baselines:
% \begin{itemize}
% \item \textbf{Greedy}: maximizes blocking effectiveness without considering fairness.
% \item \textbf{iIMDP}~\cite{Becker2023}: although proposed for IM, we adapt its LP formulation to enforce DP fairness in the IBM setting, with a time limit of 12 hours.
% \item \textbf{MMF}~\cite{Tsang2019}: maximizes the minimum utility among communities.
% \item \textbf{DC}~\cite{Tsang2019}: follows the original disparity-constrained formulation.
% \item \textbf{WF}~\cite{Rahmattalabi2021}: we set $\alpha_{\text{WF}}=0.1$ as suggested by the authors.
% \item \textbf{CFF}~\cite{Wang2025}: we adopt $\log_2(x^{0.01}+1)$ due to its reported effectiveness.
% \end{itemize}
(1) \textbf{Greedy}: maximizes blocking effectiveness without considering fairness.
(2) \textbf{iIMDP}~\cite{Becker2023}: although proposed for IM, we adapt its LP formulation to enforce DP fairness in the IBM setting, with a time limit of 12 hours.
(3) \textbf{MMF}~\cite{Tsang2019}: maximizes the minimum utility among communities.
(4) \textbf{DC}~\cite{Tsang2019}: follows the original disparity-constrained formulation.
(5) \textbf{WF}~\cite{Rahmattalabi2021}: we set $\alpha_{\text{WF}}=0.1$ as suggested by the authors.
(6) \textbf{CFF}~\cite{Wang2025}: we adopt $\log_2(x^{0.01}+1)$ due to its reported effectiveness.

The results are shown in Fig.~\ref{fig:Comparisons}.
The x-axis represents blocking effectiveness $\mathcal{F}$, and the y-axis represents fairness $\mathcal{W}$.
Points closer to the top-right corner indicate higher-quality solutions.
Our method produces a Pareto front, while all baseline methods yield a single solution that either lies on the frontier or is strictly dominated.

% As expected, Greedy occupies the extreme point with the highest blocking effectiveness.
% Since iIMDP outputs a probabilistic solution in its original formulation, we convert it into a discrete seed set by selecting the top-$k$ nodes with the highest probabilities.
% While this rounding may introduce instability on certain networks, iIMDP generally favors fairness over effectiveness.
% Moreover, despite careful tuning, iIMDP fails to terminate within the predefined 12-hour time limit on large-scale networks such as Pokec, highlighting the inherent scalability limitations of LP-based approaches.
% MMF achieves a relatively high level of fairness on datasets with a small number of communities, but its blocking effectiveness decreases significantly as the number of communities increases.
% DC performs consistently poorly across datasets, likely due to its heavy reliance on specific network structures.
% WF and CFF achieve better trade-offs than other baselines, but still fall behind the Pareto front produced by our method.
As expected, \textbf{Greedy} occupies the extreme point with the highest blocking effectiveness, serving as an upper bound when fairness is ignored.
Since \textbf{iIMDP} outputs a probabilistic solution in its original formulation, we convert it into a discrete seed set by selecting the top-$k$ nodes with the highest probabilities.
While this rounding strategy may introduce instability on certain networks, \textbf{iIMDP} generally prioritizes fairness over blocking effectiveness.
Despite careful tuning, iIMDP fails to terminate within the predefined 12-hour time limit on large-scale networks such as Pokec, reflecting the inherent scalability limitations of LP-based formulations.
\textbf{MMF} achieves relatively high fairness on datasets with a small number of communities, but its blocking effectiveness degrades substantially as the number of communities increases.
\textbf{DC} performs consistently poorly across all datasets, likely due to its strong dependence on specific internal network structures.
Both \textbf{WF} and \textbf{CFF} achieve better trade-offs than the aforementioned baselines; however, their solutions are consistently dominated by our method.

We note that there are gaps between adjacent Pareto points due to the fixed step size of $0.01$.
In practice, decision-makers can focus on regions of interest along the Pareto front and refine the search with a smaller step size to obtain solutions that better match application-specific requirements.

% Overall, these results demonstrate the superiority of CELF-R in achieving high-quality fairness–effectiveness trade-offs for the FIBM problem.
Overall, these results demonstrate the superiority of CELF-R in achieving high-quality fairness–effectiveness trade-offs.

\begin{figure*}[!t]
  \centering
  \subfigure[Facebook]
  {\includegraphics[width=0.245\linewidth]{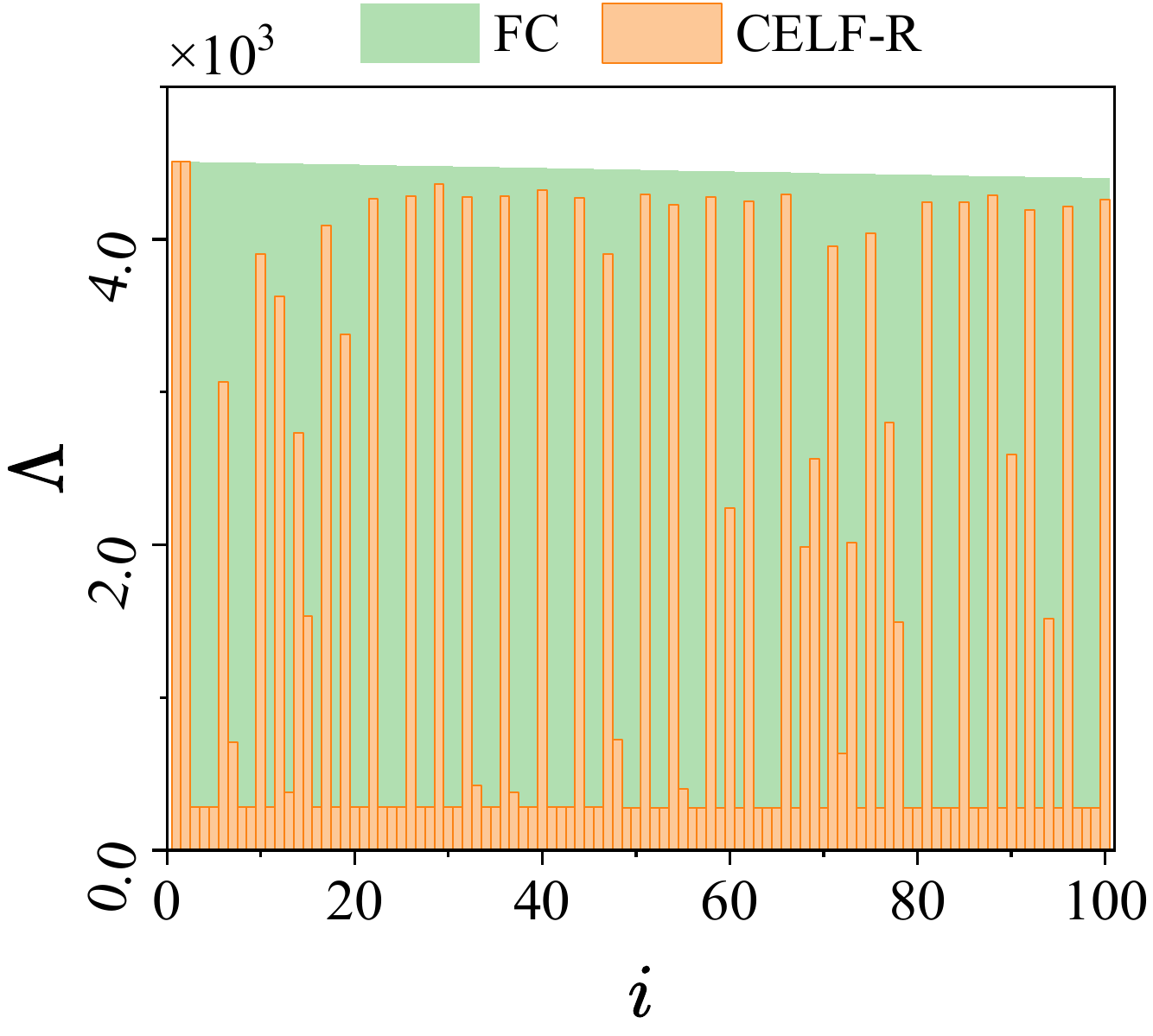}}
  \subfigure[Slashdot]
  {\includegraphics[width=0.245\linewidth]{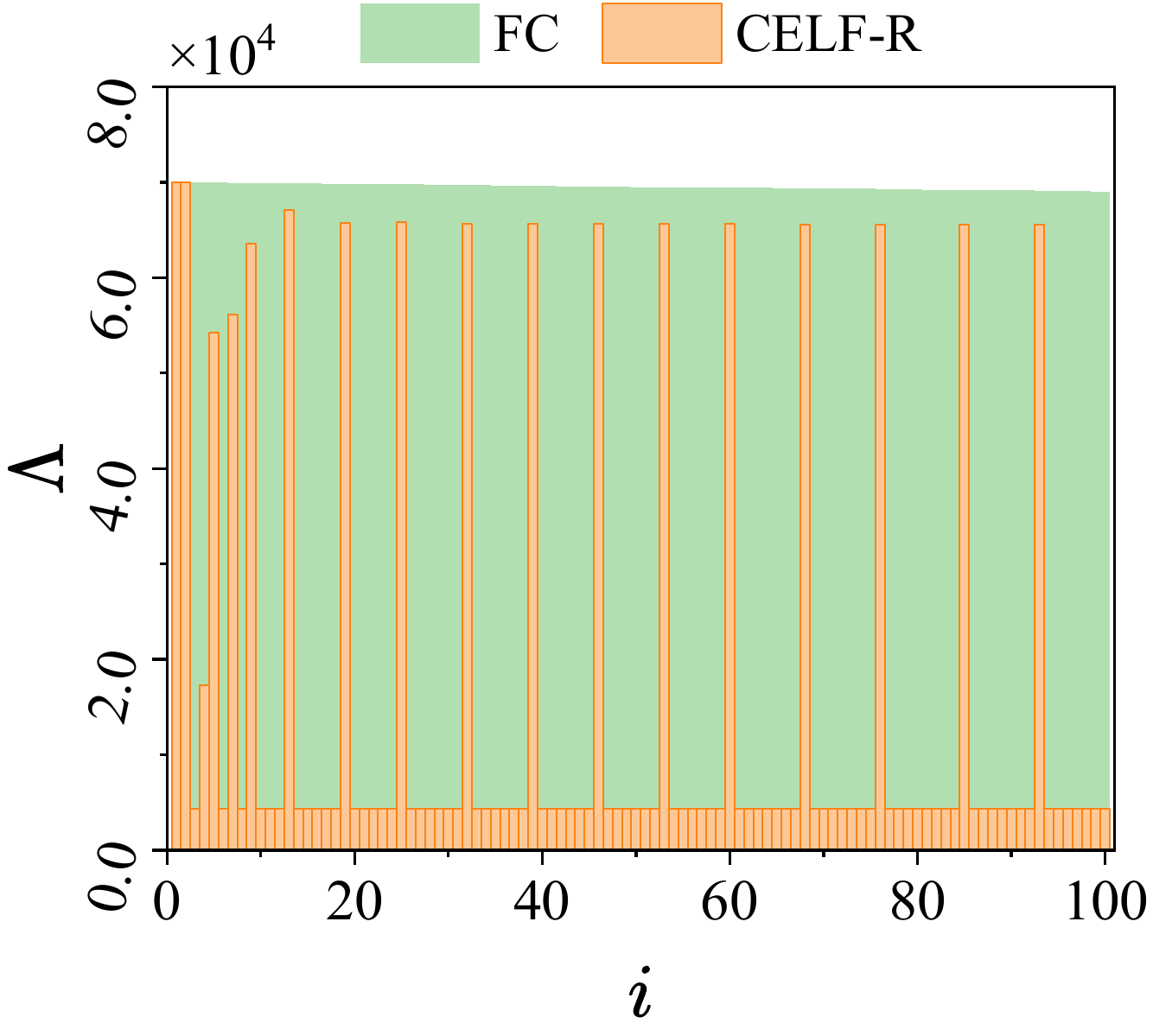}}
  \subfigure[Gowalla]
  {\includegraphics[width=0.245\linewidth]{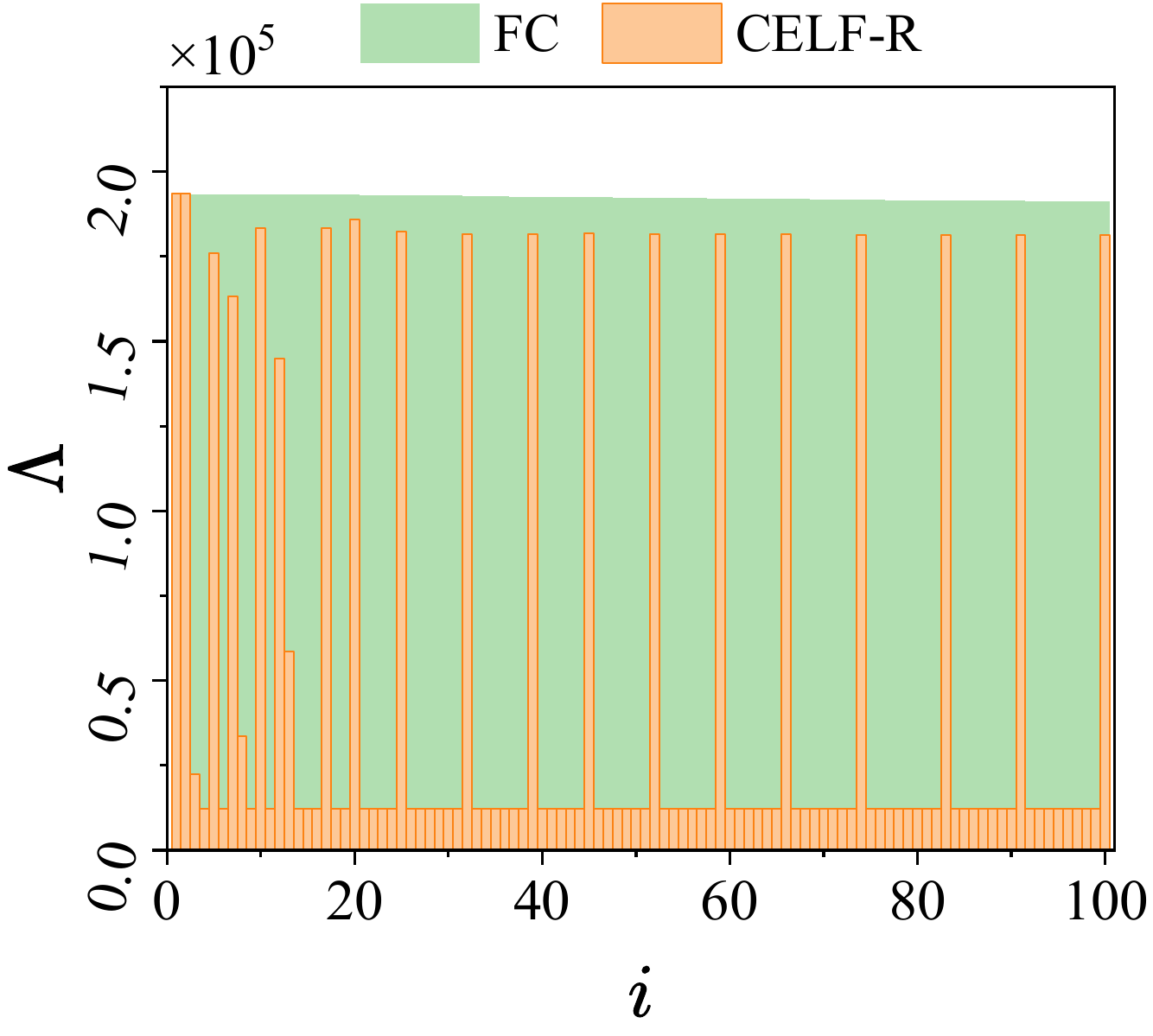}}
  \subfigure[Pokec]
  {\includegraphics[width=0.245\linewidth]{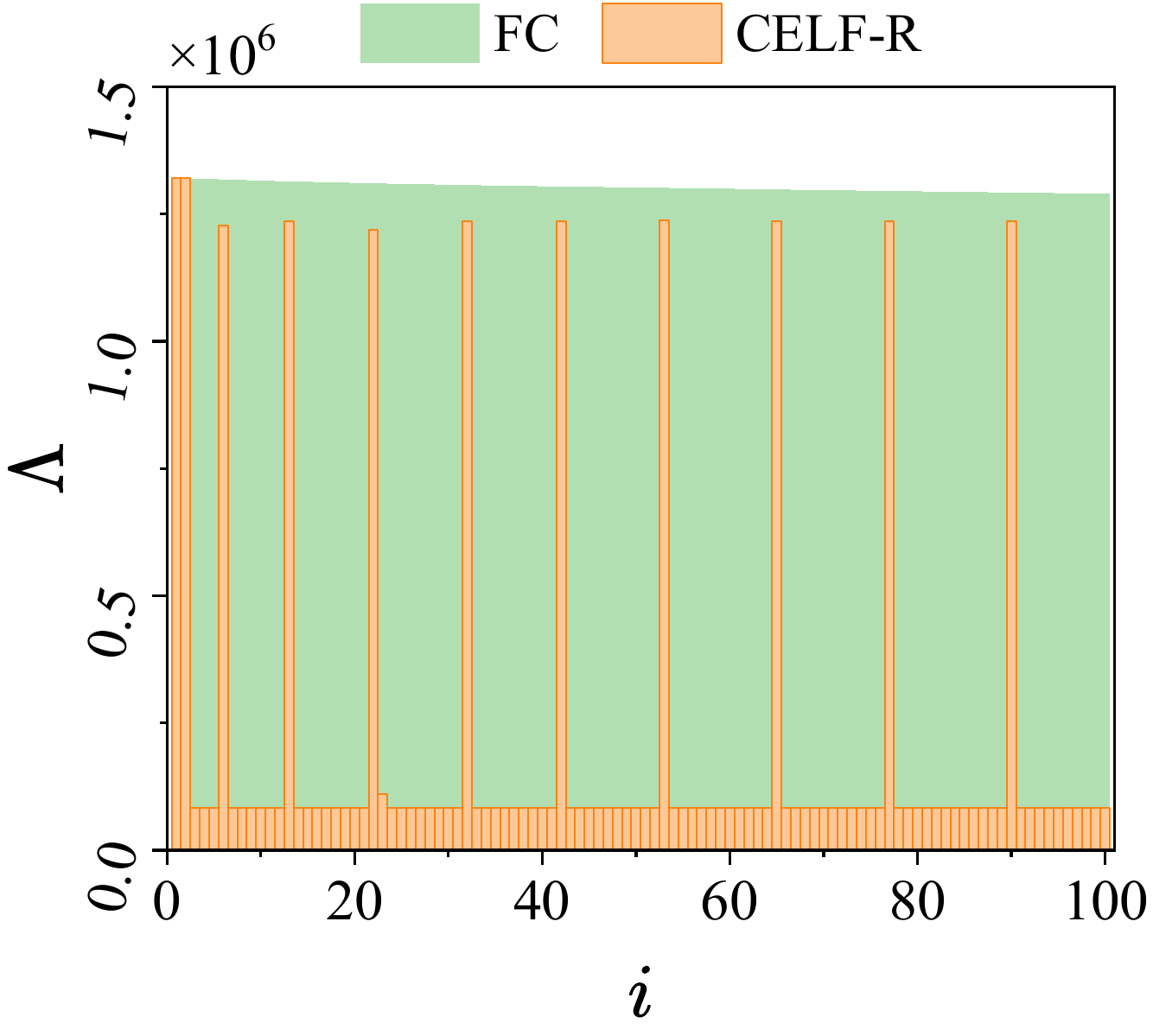}}
  % \caption{Computational efficiency of CELF-R and FC ($i$ denotes the iteration and $\Lambda$ represents the computations).}
%   \caption{Computational efficiency comparison between CELF-R and Full Computation (FC).
% The y-axis reports the total number of marginal gain evaluations $\Lambda$ at each iteration $i$.
% CELF-R consistently reduces the number of computations by more than 50\% compared to FC, demonstrating its effectiveness in eliminating redundant evaluations under approximate submodularity.}
  \caption{Computational efficiency comparison between CELF-R and Full Computation (FC).
The y-axis reports the total number of marginal gain evaluations $\Lambda$ at each iteration $i$.}

  \label{fig:computational savings}
\end{figure*}

\subsection{Seed Selection Strategy} %Ablation Study 
% To evaluate the effectiveness of CELF-R, we conduct experiments to compare its performance against two variants: CELF and FC (Full Computation).

% CELF is the original optimization strategy that assumes strict submodularity.
% FC selects the node with the highest marginal gain at each step by exhaustively examining all nodes, without leveraging any pruning operations.
% For a fair and computationally feasible comparison, all of them are parallelized and executed in the same batches.

% As shown in Figure~\ref{fig:ablation_fs}, CELF-R and FC produce nearly identical Pareto front, indicating that CELF-R achieves an accuracy comparable to exhaustively calculating marginal gains of all nodes.
% In contrast, CELF performs much worse, highlighting the importance of considering the approximate submodularity deviation $\epsilon$ in seed selection.
To evaluate the effectiveness of CELF-R, we conduct an ablation study comparing it against two variants: \textbf{CELF} and \textbf{FC} (Full Computation).
CELF is the original lazy greedy strategy that assumes strict submodularity.
FC selects the node with the highest marginal gain at each iteration by exhaustively evaluating all candidates, without any pruning or lazy updates.
For a fair and computationally feasible comparison, all methods are parallelized and executed using the same batch configurations.

As shown in Fig.~\ref{fig:ablation_fs}, CELF-R and FC produce nearly identical Pareto fronts, indicating that CELF-R achieves solution quality comparable to full marginal gain recomputation.
In contrast, CELF performs significantly worse, highlighting the necessity of explicitly accounting for the approximate submodularity deviation $\epsilon$ during seed selection.

\subsection{Efficiency Evaluation}
% To evaluate the computational efficiency gained through CELF-R, we compare the number of marginal gain computations required by CELF-R and FC.
% Fig.~\ref{fig:computational savings} exhibits the results where the green bars indicate the total number of computations required by FC, while the orange bars show those required by CELF-R.
% The results clearly demonstrate that CELF-R achieves substantial computational savings, reducing the number of evaluations by an average of 53.60\% to 56.43\%.
% This confirms the effectiveness of CELF-R in eliminating redundant computations, validating its design for efficient seed selection under approximate submodularity.

% In addition, we compare the runtime of iIMDP.
% Table~\ref{Tab:runtime} illustrates the runtime across all datasets where CELF-R consistently outperforms FC in efficiency, achieving a speedup of $48.2\%$ to $60.2\%$. 
% And extremely faster than the LP-based methods.
To quantify the computational efficiency gained by CELF-R, we compare the number of marginal gain evaluations required by CELF-R and FC.
Fig.~\ref{fig:computational savings} reports the results, where the green bars denote the total number of evaluations performed by FC and the orange bars correspond to CELF-R.
% CELF-R consistently achieves substantial computational savings, reducing the number of evaluations by an average of 53.60\% to 56.43\%.
These results confirm that CELF-R effectively eliminates redundant computations while preserving solution quality.

We further compare runtime performance against iIMDP.
Table~\ref{Tab:runtime} summarizes the results across all datasets.
CELF-R consistently outperforms FC, achieving speedups ranging from 48.2\% to 60.2\%, and is orders of magnitude faster than LP-based methods.
These results demonstrate that CELF-R makes fair influence blocking maximization practical on large-scale networks.

\begin{table}[!ht]
    \centering
    \caption{Runtime Comparisons ($k=100$)}
    \label{Tab:runtime}
    \begin{tabular}{ccccc}
    \toprule
    \textbf{Dataset} & \textbf{iIMDP}  & \textbf{CELF-R} & \textbf{FC} & \%Improv.\\
    \midrule
     Facebook &$5.13$s &$1.07$s &$2.15$s &$50.2$\%\\
     Slashdot &$1.1$h &$4.43$s &$11.07$s &$60.2$\%\\
     Gowalla &$6.5$h &$15.21$s &$32.97$s &$53.9$\%\\
     Pokec    &$>12$h  &$49.65$s  &$95.85$s &$48.2$\%\\
    \bottomrule
    \end{tabular}
\end{table}

\subsection{Empirical Bound}
% Theorem~\ref{therorem: K} provides conservative approximation guarantees for $K(\cdot)$. 
% Even so, in practical applications, we employ a dynamic update strategy, which typically results in a tighter empirical bound.
% Let $\psi_k$ denote the performance loss under budget constraint $k$, where $\psi_k$ is calculated by the accumulation of $\epsilon$ and $\kappa$ across all iterations, leading to $K(S_{P_{b^*}})\geq(1-1/e)K(S^*_P)-\psi_k$.
% Table~\ref{Tab:Estimator} shows the specific $\psi_k$ values on all datasets.
% In practice, its value varies a lot between different networks and seems highly relevant to the graph structure.
% Overall, on these four networks, it indicates that as the budget $k$ increases, $\psi_k$ rises slightly, with the overall performance loss remaining within an acceptable range.

Theorem~\ref{therorem: K} provides conservative worst-case approximation guarantees for $\mathcal{K}(\cdot)$, which assumes fixed upper bounds on deviation of $\epsilon_\text{max}$ and $\kappa_\text{max}$ at all times.
In practice, however, CELF-R employs a dynamic deviation tracking strategy that grows slightly towards these upper bounds, which typically yields much tighter empirical bounds.
Let $\psi_k$ denote the empirical bound under budget constraint $k$, computed from the cumulative deviations $\epsilon$ and $\kappa$ across all iterations.
% Let $\psi_k$ denote the accumulated performance loss under budget constraint $k$, computed from the cumulative deviations $\epsilon$ and $\kappa$ across all iterations.
% This leads to the empirical bound
% $\mathcal{K}(S_{P}) \ge (1 - 1/e)\mathcal{K}(S_P^*) - \psi_k.$

Table~\ref{Tab:Estimator} reports the observed $\psi_k$ values across all datasets.
While $\psi_k$ varies across networks, likely due to differences in graph structure, it remains relatively small in all cases.
% As the budget $k$ increases, $\psi_k$ grows moderately, indicating that the practical performance loss induced by approximate submodularity remains well controlled.
As the budget $k$ increases, $\psi_k$ grows moderately, indicating that the practical performance loss remains well controlled.

\begin{table}[!ht]
    \centering
    % \small
    \caption{$\psi_k$ under the constraint of budget $k$.}
    \label{Tab:Estimator}
    \begin{tabular}{cccccc}
    \toprule
    \textbf{Dataset} & \textbf{$\psi_{20}$} & \textbf{$\psi_{40}$} & \textbf{$\psi_{60}$}& \textbf{$\psi_{80}$} & \textbf{$\psi_{100}$} \\
    \midrule
     Facebook &0.018 &0.032 &0.058 &0.078 &0.095\\
     Slashdot &0.014 &0.033 &0.039 &0.055 &0.063\\
     Gowalla &0.011 &0.033 &0.052 &0.077 &0.102\\
     Pokec  &0.003 &0.007 &0.015 &0.023 &0.027\\
    \bottomrule
    \end{tabular}
\end{table}

%% file: 7Conclusion.tex
\section{Conclusion}
In this paper, we incorporate Demographic Parity (DP) into Influence Blocking Maximization and formalize the Fair Influence Blocking Maximization (FIBM) problem.
To overcome the reliance of traditional DP formulations on costly solvers, we propose a DP-aware objective that admits an approximately monotonic submodular structure.
Combined with blocking effectiveness via a tunable linear scalarization, this formulation enables efficient optimization and flexible control of fairness–effectiveness trade-offs with theoretical guarantees.
We further develop CELF-R, an efficient seed selection strategy tailored to this approximate structure, which naturally supports Pareto front construction under different preferences.
Experiments on four real-world networks demonstrate that CELF-R consistently outperforms state-of-the-art baselines, achieving a $(1-1/e-\psi)$-approximate solution while maintaining high efficiency.

%% file: 8Appendix.tex
\newpage

\appendix
\section{Appendix}
\subsection{Proofs of Lemmas and Theorems}

\lemmaWDP*
\begin{proof}
Since $x_c(S_P) =\frac{\sigma^-_c(S_P)}{\sigma^-(S_P)}$, it holds $\sum_{c\in\cC}x_c(S_P) = 1$.
Therefore, we construct the Lagrangian function as
\begin{equation}
    \Gamma(x_{c_1}, \cdots, x_{c_n}) = \sum_{c \in \cC} r_c x_c^\alpha + \lambda \left(1 - \sum_{c \in \cC} x_c\right).
\end{equation}

% %%% new contents
Taking the partial derivative for $x_c$, it obtains
\begin{equation}
    \frac{\partial \Gamma}{\partial x_c} = \alpha r_c x_c^{\alpha - 1} - \lambda.
\end{equation}

If we set $\lambda$ sufficiently large, then $\frac{\partial \Gamma}{\partial x_c} \leq 0$ holds for any $c\in \cC$. 
Therefore, $\mathcal{W}(\cdot)$ would be maximized when $\frac{\partial \Gamma}{\partial x_c} = 0$ for all $c\in \cC$.
In such a case, we have
\begin{equation}
   r_{c_1} x_{c_1}^{\alpha - 1} = \cdots = r_{c_n} x_{c_n}^{\alpha - 1} = \frac{\lambda}{\alpha}.\label{eq:equal_rn}
\end{equation}

Substituting $r_c = n_c^{1 - \alpha}$ into Eq.~(\ref{eq:equal_rn}) gives
\begin{equation}
    \left(\frac{x_{c_1}}{n_{c_1}}\right)^{\alpha - 1} = \cdots = \left(\frac{x_{c_n}}{n_{c_n}}\right)^{\alpha - 1}.
\end{equation}

Since $\alpha - 1 \neq 0$ and $\sum_{c\in\cC}x_c>0$, we have
\begin{equation}
    \frac{x_{c_1}}{n_{c_1}} = \cdots = \frac{x_{c_n}}{n_{c_n}} \neq 0.
\end{equation}

Recalling the definitions of $n_c$ and $x_c$ as
\begin{equation}
    n_c = \frac{\sigma_c(S_N, G)}{\sigma(S_N, G)}, \quad x_c = \frac{\sigma^-_c(S_P)}{\sigma^-(S_P)}.
\end{equation}

Hence, the proportion of blocked nodes relative to community size is equal across all communities:
\begin{equation}
    \frac{\sigma^-_{c_1}(S_P)}{\sigma_{c_1}(S_N,G)} = \cdots = \frac{\sigma^-_{c_n}(S_P)}{\sigma_{c_n}(S_N,G)}.
\end{equation}

This satisfies the definition of strict Demographic Parity (DP), which requires that the utility gap across all communities be exactly zero.
Thus, maximizing $\mathcal{W}(S_P)$ naturally induces a fair blocking allocation that obeys DP, which concludes the proof.
\end{proof}

\lemmaAlpha*
\begin{proof}
Let $h=\frac{x_c}{n_c} - 1$.
Recall that $\mathcal{W}_\alpha=\sum_{c \in \cC}n_c^{1-\alpha}x_c^\alpha$. 
Substituting $\frac{x_c}{n_c} - 1$ with $h$ gives
\begin{equation}
    \mathcal{W}_\alpha = \sum_{c\in \cC}n_c^{1-\alpha}n_c^\alpha(h+1)^\alpha = \sum_{c \in \cC}n_c(h+1)^\alpha.
\end{equation}

Applying Taylor expansion, we have
\begin{equation}
    (h + 1)^\alpha = 1 + \alpha h + \frac{\alpha (\alpha-1)}{2} h^2 + O(h^3).
\end{equation}

Since $\sum_{c\in \cC}n_c=\sum_{c\in \cC}x_c=1$, we obtain
\begin{equation}
    \sum_{c\in \cC}n_ch=\sum_{c\in \cC}x_c-\sum_{c \in \cC}n_c=0.
\end{equation}

Hence,
\begin{small}
\begin{equation}
    \begin{aligned}
    \mathcal{W}_\alpha &= \sum_{c\in \cC}n_c + \alpha\sum_{c\in \cC}n_ch + \frac{\alpha(\alpha-1)}{2}\sum_{c\in \cC}n_ch^2 + O(\sum_{c\in \cC}n_ch^3) \\
    &= 1 + \frac{\alpha(\alpha-1)}{2}\sum_{c\in \cC}n_ch^2 + O(\sum_{c\in \cC}n_ch^3).
    \end{aligned}
\end{equation}
\end{small}

For any $\alpha_1,\alpha_2\in(0,1)$, we have
\begin{equation}
\begin{aligned}
    |\mathcal{W}_{\alpha_1}-\mathcal{W}_{\alpha_2}| &= |\frac{\pi(\alpha)}{2}\sum_{c\in \cC}n_ch^2 + O(\sum_{c\in \cC}n_ch^3)| \\
    &\leq \frac{|\pi(\alpha)|}{2}\sum_{c\in \cC}n_ch^2 + O(\sum_{c\in \cC}n_c|h^3|),
\end{aligned}
\end{equation}
where $\pi(\alpha) = (\alpha_1-\alpha_2)(\alpha_1+\alpha_2-1)$.

Since $|h| \leq \phi$, we have
\begin{equation}
\begin{aligned}
    &\sum_{c \in \cC}n_ch^2 \leq \phi^2\sum_{c \in \cC}n_c = \phi^2. \\
    &\sum_{c \in \cC}n_c|h^3| \leq \phi^3\sum_{c \in \cC}n_c = \phi^3.
\end{aligned}
\end{equation}

Therefore,
\begin{equation}
     |\mathcal{W}_{\alpha_1}-\mathcal{W}_{\alpha2}| \leq \frac{|(\alpha_1-\alpha_2)(\alpha_1+\alpha_2-1)|}{2}\phi^2+O(\phi^3),
\end{equation}
which concludes the proof.
\end{proof}

\lemmaW*
\begin{proof}
\textbf{Approximate Submodularity}. For any subset $X \subseteq Y \subseteq V$ and element $u \notin Y$,  
\[
\mathcal{W}(X \cup \{u\}) - \mathcal{W}(X) \geq \mathcal{W}(Y \cup \{u\}) - \mathcal{W}(Y) - \epsilon.
\]
where $\epsilon$ quantifies the deviation from strict submodularity.

Let $x_c(S) = n_c + \delta_c(S)$, where $\delta_c(S)$ denotes the difference between $n_c$ and $x_c$. 
Then, the marginal gain of adding $u$ to $S$ is
\begin{equation}
\begin{aligned}
\Delta_S(u) &= \mathcal{W}(S \cup \{u\}) - \mathcal{W}(S) \\
&= \sum_{c \in \cC} n_c^{1 - \alpha} \big[ \big( n_c + \delta_c(S\cup\{u\}) \big)^\alpha - \big(x_c(S)\big)^\alpha \big] \\
&= \sum_{c \in \cC} n_c^{1 - \alpha} \big[ \left(x_c(S) +\delta_c^u(S) \big)^\alpha - \big(x_c(S)\big)^\alpha \right],
\end{aligned}
\end{equation}
where $\delta_c^u(S)$ denotes the additional influence brought by $u$ to community $c$.

Applying the Taylor expansion to $\big(x_c(S) + \delta_c^u(S) \big)^\alpha$, we get
\begin{equation}
\begin{aligned}
\begin{small}
\big(x_c(S) + \delta_c^u(S) \big)^\alpha = \big(x_c(S)\big)^\alpha + \alpha \big(x_c(S)\big)^{\alpha - 1} \delta_c^u(S) + R_c(S),
\end{small}
\end{aligned}
\end{equation}
with $R_c(S)$ given by the Lagrange form as follows:
\begin{equation}
\begin{aligned}
R_c(S) = \frac{\alpha(\alpha-1)\big(\xi_c\big)^{\alpha-2} }{2} \big(\delta_c^u(S)\big)^2, \xi_c \in (\xi_c^{min}, \xi_c^{max}),
\end{aligned}
\end{equation}
where $\xi_c^{min}=\min \big(x_c(S), x_c(S)+\delta_c^u(S) \big)$, and $\xi_c^{max}=\max \big(x_c(S), x_c(S)+\delta_c^u(S) \big)$.

Hence, we have
\begin{equation}
|R_c(S)| < \frac{\alpha(1 - \alpha)}{2}(\xi_c^{min})^{\alpha-2}\big(\delta_c^u(S)\big)^2,
\end{equation}
which leads to
\begin{equation}
\begin{aligned}
\Delta_S(u) &= \sum_{c \in \cC} n_c^{1 - \alpha} \big[ \alpha \big(x_c(S)\big)^{\alpha - 1} \delta_c^u(S) + R_c(S) \big].
\end{aligned}
\end{equation}

Again, applying the Taylor expansion to $\big(x_c(S)\big)^{\alpha - 1}$, {\em i.e.}, $ \big(n_c + \delta_c(S)\big)^{\alpha - 1}$, we get
\begin{equation}
\begin{small}
\big(n_c + \delta_c(S)\big)^{\alpha - 1} = n_c^{\alpha - 1} + (\alpha - 1) n_c^{\alpha - 2} \delta_c(S) + R_c'(S),
\end{small}
\end{equation}
with $R_c'(S)$ formulated as
\begin{equation}
\begin{small}
R_c'(S) = \frac{(\alpha-1)(\alpha-2)\big(\eta_c\big)^{\alpha-3} }{2} \big(\delta_c(S)\big)^2,
\end{small}
\end{equation}
where $\eta_c\in(\eta_c^{min},\eta_c^{max})$, $\eta_c^{min}= \min\big(n_c, n_c + \delta_c(S)\big)$, and $\eta_c^{max}= \max\big(n_c, n_c + \delta_c(S)\big)$.

Hence, we have
\begin{equation}
|R_c'(S)| \leq \frac{(\alpha-1)(\alpha-2)}{2} (\eta_c^{min})^{\alpha-3} \big(\delta_c(S)\big)^2.
\end{equation}

Substituting the above into $\Delta_S(u)$ leads to
\begin{equation}
\begin{aligned}
\Delta_S(u) &= \sum_{c \in \cC} n_c^{1 - \alpha} \big\{ \alpha \big(x_c(S)\big)^{\alpha - 1} \delta_c^u(S) + R_c(S) \big\} \\
&= \sum_{c \in \cC} n_c^{1-\alpha} \big\{ \alpha \big[ n_c^{\alpha - 1} + (\alpha - 1) n_c^{\alpha - 2} \delta_c(S) \\
& \quad\ + R_c'(S) \big] \delta_c^u(S) + R_c(S) \big\} \\
&= \sum_{c \in \cC} \big\{ \alpha \delta_c^u(S) + \frac{\alpha (\alpha - 1) \delta_c(S) \delta_c^u(S)}{n_c}  \\ 
& \quad\ +\underbrace{\alpha n_c^{1-\alpha} \delta_c^u(S) R_c'(S) + n_c^{1-\alpha} R_c(S)}_{\text{Remainder } E_c(S)} \big\}.
\end{aligned}
\end{equation}

The remainder $E_c(S)$ satisfies
\begin{equation}
\begin{aligned}
    |E_c(S)| 
    &= O\left( |\delta_c^u(S)| \cdot |\delta_c(S)|^2 + |\delta_c^u(S)|^2 \right) \\
    &\ll \left| \alpha \delta_c^u(S) + \frac{\alpha (\alpha - 1) \delta_c(S) \delta_c^u(S)}{n_c} \right|.
\end{aligned}
\end{equation}

Since $E_c(S)$ is negligible, we obtain
\begin{equation}
\Delta_S(u) = \alpha \sum_{c \in \cC} \left[ 1 + \frac{(\alpha - 1) \delta_c(S)}{n_c} \right] \delta_c^u(S).
\end{equation}

As $\sum_{c\in\cC} n_c=\sum_{c\in\cC} x_c = 1$, it holds $\sum_{c\in\cC} \delta_c^u(S)=0$. 
Therefore, we have
\begin{equation}
\Delta_S(u) = \alpha (\alpha - 1) \sum_{c \in \cC}\frac{\delta_c(S) \delta_c^u(S)}{n_c}.\label{eq:deltaT}
\end{equation}

For any subset $X \subseteq Y \subseteq V$ and element $u \notin Y$, the difference in marginal gains becomes
\begin{equation}
    \Delta_X(u) - \Delta_Y(u) = \alpha (\alpha - 1) \sum_{c \in \cC} \frac{\delta_c(X) \delta_c^u(X) - \delta_c(Y) \delta_c^u(Y)}{n_c}.
\end{equation}

% Assuming $ |\delta_c(S)| \leq \delta_{\text{max}} $ and $ |\delta_c^u(S)| \leq \delta_{\text{max}}' $, we bound the difference as:
Let $\delta_{\text{max}} = \max_{c\in\cC} |\delta_c(S)|$ and $\delta^u_{\text{max}} = \max_{c\in\cC} |\delta_c^u(S)|$, then we can bound the different as follows:
\begin{equation}
    |\Delta_X(u) - \Delta_Y(u)| \leq 2 \alpha (1 - \alpha) \delta_{\text{max}} \delta^u_{\text{max}} \sum_{c \in \cC} \frac{1}{n_c}.
\end{equation}

Thus, we have
\begin{equation}
\epsilon = 2 \alpha (1 - \alpha) \delta_{\text{max}} \delta^u_{\text{max}} \sum_{c \in \cC} \frac{1}{n_c}.
\end{equation}

\textbf{Approximate Monotonicity}. For any subset $X \subseteq V$ and element $u \notin X$,  
\[
\mathcal{W}(X \cup \{u\}) \geq \mathcal{W}(X) - \kappa.
\]

Similar to the proof of approximate submodularity, we have
\begin{equation}
\begin{aligned}
\mathcal{W}(S \cup \{u\}) - \mathcal{W}(S) &= \alpha (\alpha - 1) \sum_{c \in \cC}\frac{\delta_c(S) \delta_c^u(S)}{n_c} \\
&\geq \alpha(\alpha - 1)\sum_{c \in \cC}\frac{\delta_{\text{max}}\delta^u_{\text{max}}}{n_c}.
\end{aligned}
\end{equation}

Thus, we have
\begin{equation}
\kappa=\alpha(1-\alpha)\delta_{\text{max}}\delta^u_{\text{max}}\sum_{c \in \cC}\frac{1}{n_c}.
\end{equation}

Therefore, $\mathcal{W}(\cdot)$ satisfies both approximate submodularity and approximate monotonicity, which concludes the proof.
\end{proof}

\lemmaopt*
\begin{proof}
Let $\chi^*$ be an optimal solution with $|\chi^*|=|\chi|$, {\em i.e.}, $\mathcal{K}(\chi^*)=OPT$, we denote the elements in $\chi \backslash \chi^*$ by $u_1^*, u_2^*, ..., u_l^*$ and the elements in $\chi^* \backslash \chi$ by $v_1^*, v_2^*, ..., v_l^*$, where $|\chi \backslash \chi^*|=|\chi^* \backslash \chi|=l \leq k$.

Since $\mathcal{K}(\cdot)$ satisfies approximate monotonicity, we have
\begin{equation}
\begin{aligned}
    \mathcal{K}(\chi^* \cup \chi) &= \mathcal{K}(\chi^* \cup \{u_1^*, u_2^*, ..., u_l^*\}) \\
    &\geq \mathcal{K}(\chi^* \cup \{u_1^*, u_2^*, ..., u_{l-1}^*\}) - \kappa \\
    &\geq \mathcal{K}(\chi^*) - l\kappa.
\end{aligned}
\end{equation}

Combining approximate submodularity, we obtain
\begin{equation}
\begin{aligned}
    \mathcal{K}(\chi^*) &- \mathcal{K}(\chi) \\
    &\leq \mathcal{K}(\chi \cup \chi^*) - \mathcal{K}(\chi) + l\kappa\\
     &= \mathcal{K}(\chi \cup \{v_1^*, v_2^*, ..., v_l^*\}) - \mathcal{K}(\chi) + l\kappa \\
     &= \sum_{i=1}^l\big(\mathcal{K}(\chi \cup \{v_1^*, v_2^*, ..., v_i^*\}) \\ &- \mathcal{K}(\chi \cup \{v_1^*, v_2^*, ..., v_{i-1}^*\})\big) + l\kappa \\
     &\leq \sum_{i=1}^l\big(\mathcal{K}(\chi \cup \{v_i^*\}) - \mathcal{K}(\chi) + \epsilon\big) + l\kappa.
\end{aligned}
\end{equation}

Therefore, for a specific element $v^*$ that satisfies the condition $v\notin\chi$, it holds
\begin{equation}\label{eq:approx-asm}
\begin{aligned}
    \mathcal{K}(\chi \cup \{v^*\}) &- \mathcal{K}(\chi) \\ 
     &\geq \frac{1}{l}\big(\mathcal{K}(\chi^*)-\mathcal{K}(\chi) - l\kappa\big) - \epsilon \\
     &\geq \frac{1}{l}\big(\mathcal{K}(\chi^*)-\mathcal{K}(\chi)\big)-\kappa -\epsilon \\
     &\geq \frac{1}{k}\big(OPT- \mathcal{K}(\chi)\big) - \kappa - \epsilon.
\end{aligned}
\end{equation}
which concludes the proof.
\end{proof}

\theoremK*
\begin{proof}
Based on Lemma~\ref{lemma:opt}, we have
\begin{equation}\label{eq:opt-greedy}
\begin{aligned}
    \mathcal{K}(\chi^*) &- \mathcal{K}(\chi \cup \{v^*\}) \\ 
    &\leq \mathcal{K}(\chi^*) - (1-\frac{1}{k})\mathcal{K}(\chi) - \frac{1}{k}\big(\mathcal{K}(\chi^*) - k\kappa - k\epsilon\big) \\
    &= (1-\frac{1}{k})\big(\mathcal{K}(\chi^*) - \mathcal{K}(\chi)\big) + \kappa + \epsilon.
\end{aligned}
\end{equation}

Let $\Delta_i = \mathcal{K}(\chi^*) - \mathcal{K}(\chi_i)$, where $\chi_i$ is the greedy solution after $i$ steps.
Then the above Eq.~\eqref{eq:opt-greedy} implies
\begin{equation}
    \Delta_{i+1} \leq (1 - \frac{1}{k})\Delta_i + \kappa + \epsilon,
\end{equation}
which leads to
\begin{equation}
    \Delta_k \leq (1 - \frac{1}{k})^k\mathcal{K}(\chi^*) + (\kappa + \epsilon)k[1-(1 - \frac{1}{k})^k].
\end{equation}

Then we have
\begin{equation}
\begin{aligned}
    \mathcal{K}(\chi) &= \mathcal{K}(\chi^*) - \Delta_k \\
    &\geq [1 - (1 - \frac{1}{k})^k][\mathcal{K}(\chi^*) - (\kappa + \epsilon)k] \\
    &\geq (1 - 1/e) \cdot \big(\mathcal{K}(\chi^*) - k\kappa - k\epsilon \big).
\end{aligned}
\end{equation}

Let $\psi = (1 - 1/e)\cdot(k\kappa + k\epsilon)$.
Thus
\begin{equation}
    \mathcal{K}(\chi) \ge (1-1/e) \cdot \mathcal{K}(\chi^*) - \psi,
\end{equation}
which concludes the proof.
\end{proof}